%% file: implicit-fluids-amsart.tex
\documentclass[reqno, a4paper]{amsart}
\usepackage{amsmath}
\usepackage{amssymb}
\usepackage{amsthm}

\usepackage[scale=0.8]{geometry}

\usepackage{natbib}
\usepackage{bibentry} 

\usepackage[english]{babel}
\usepackage[utf8]{inputenc}

\usepackage{subfig}
\usepackage{graphicx}

\usepackage[unicode]{hyperref}
\usepackage[active]{srcltx} 

\usepackage{mathbbol}
\usepackage{bm} 
\usepackage{MnSymbol} 
\usepackage{gensymb}
\usepackage{eurosym}

\usepackage{units}
\usepackage{tensor}
\usepackage{accents}

\usepackage{enumitem}

\usepackage{lineno}

\usepackage{multicol}

\usepackage{comment}

\usepackage{placeins}

\usepackage{xcolor}
\usepackage{mdframed}
\usepackage{newfloat}

\title[Viscoelastic fluids]{Implicit type constitutive relations for elastic solids and their use in the development of mathematical models for viscoelastic fluids}

\author{V\'{\i}t Pr\r{u}\v{s}a}
\date{\today}
\address{
Faculty of Mathematics and Physics\\
Charles University\\
Sokolovsk\'a 83\\
Praha 8 -- Karl\'{\i}n\\
CZ 186\;75\\
Czech Republic
}
\email{prusv@karlin.mff.cuni.cz}

\author{K. R. Rajagopal}
\address{
Texas A\&M University\\
Department of Mechanical Engineering\\
3123 TAMU\\
College Station TX 77843-3123\\
United States of America
}
\email{krajagopal@tamu.edu}

\thanks{V\'{\i}t Pr\r{u}\v{s}a thanks the Czech Science Foundation for the support through the project 20-11027X. K.~R.~Rajagopal thanks the Office of Naval Research for its support.}

\input{vit-prusa-macros-experimental}

\input{vit-prusa-macros-reference-natural}

\input{vit-prusa-env-list}

\input{summary-style}

\let\cite\citet
\numberwithin{equation}{section}
\renewcommand{\tensordot}[2]{\ensuremath{#1 \bullet #2}}
\newtheorem{Theorem}{Theorem} 
\newtheorem{Lemma}[theorem]{Lemma} 

\begin{document}

\begin{abstract}
    Viscoelastic fluids are non-Newtonian fluids that exhibit both ``viscous'' and ``elastic'' characteristics in virtue of mechanisms to store energy and produce entropy. Usually the energy storage properties of such fluids are modelled using the same concepts as in the classical theory of nonlinear solids. Recently new models for elastic solids have been successfully developed by appealing to implicit constitutive relations, and these new models offer a different perspective to the old topic of elastic response of materials. In particular, a sub-class of implicit constitutive relations, namely relations wherein the left Cauchy-Green tensor is expressed as a function of stress is of interest. We show how to use this new perspective it the development of mathematical models for viscoelastic fluids, and we provide a discussion of the thermodynamic underpinnings of such models. We focus on the use of Gibbs free energy instead of the Helmholtz free energy, and using the standard Giesekus/Oldroyd-B models, we show how the alternative approach works in the case of well-known models. The proposed approach is straightforward to generalise to more complex setting wherein the classical approach might be impractical of even inapplicable.
\end{abstract}

\maketitle


\input{text/implicit-fluids-body}



\bibliographystyle{chicago}
\bibliography{vit-prusa}

\end{document}

%% file: vit-prusa-macros-experimental.tex
\DeclareMathOperator{\divergence}{div}

\DeclareMathOperator{\Tr}{Tr}

\newcommand{\exponential}[1]{\ensuremath{{\mathrm e}^{#1}}}

\newcommand{\reference}{\mathrm{ref}}

\newcommand{\bydefinition}{\mathrm{def}}
\newcommand{\traceless}[1]{{#1}_{\delta}}

\newcommand{\diff}{\mathrm{d}}
\newcommand{\Diff}[1][]{\mathrm{D}_{#1}} 

\renewcommand{\vec}[1]{\ensuremath{\mathbf{#1}}}

\makeatletter
\@ifpackageloaded{bm}%
{\renewcommand{\vec}[1]{\ensuremath{\bm{#1}}}%
}{%
\relax
}
\makeatother
\newcommand{\tensorq}[1]{\ensuremath{\mathbb{#1}}}      
\newcommand{\tensorc}[1]{\ensuremath{\mathrm{#1}}}      

\newcommand{\transpose}[1]{#1^\top}

\newcommand{\inverse}[1]{#1^{-1}}

\newcommand{\identity}{\ensuremath{\tensorq{I}}} 
\newcommand{\tensorzero}{{\mathbb{O}}} 

\newcommand{\cstress}{\tensorq{T}}

\newcommand{\thcstress}{\ensuremath{\cstress_{\mathrm{th}}}} 

\newcommand{\fgrad}{\tensorq{F}}

\newcommand{\rcg}{\tensorq{C}}
        
\newcommand{\rcgrel}[3][]{\rcg^{#1}_{#2}\left(#3\right)}

\newcommand{\lcg}{\tensorq{B}}

\makeatletter
\@ifpackageloaded{bm}%
{%
}{%

}

\@ifpackageloaded{bm}%
{%
 
}{%

}

\@ifpackageloaded{bm}%
{%
}{%

}
\makeatother

\newcommand{\henckystrain}{\tensorq{H}} 

\newcommand{\generictensor}{{\tensorq{A}}}
\newcommand{\generictensorc}{\tensorc{A}} 

\newcommand{\dcstresssymb}{\traceless{\cstress}}

\newcommand{\gradsym}{\ensuremath{\tensorq{D}}}

\newcommand{\gradvl}{\ensuremath{\tensorq{L}}}

\newcommand{\functional}[1]{{\mathfrak #1}}
\newcommand{\fhistory}[3]{{\functional{#1}_{#2}^{#3}}}

\newcommand{\kdelta}[1]{\tensor{\delta}{#1}}

\newcommand{\R}{\ensuremath{{\mathbb R}}}

\newcommand{\ienergy}{\ensuremath{e}} 
\newcommand{\fenergy}{\ensuremath{\psi}} 
\newcommand{\entropy}{\ensuremath{\eta}} 
\newcommand{\gibbs}{\ensuremath{g}} 

\newcommand{\temp}{\ensuremath{\theta}} 
\newcommand{\thpressure}{\ensuremath{p_{\mathrm{th}}}} 

\newcommand{\mns}{\ensuremath{m}} 

\newcommand{\cheatvolref}{\ensuremath{c_{\mathrm{V}, \reference}}} 

\newcommand{\hfluxc}{\vec{j}_{q}}     

\newcommand{\entfluxc}{\vec{j}_{\entropy}} 

\newcommand{\entprodc}{\xi} 

\newcommand{\pd}[2]{\ensuremath{\frac{\partial {#1}}{\partial {#2}}}}
\newcommand{\ppd}[2]{\ensuremath{\frac{\partial^2 {#1}}{\partial {#2^2}}}}
\newcommand{\dd}[2]{\ensuremath{\frac{\diff {#1}}{\diff {#2}}}}

\newcommand{\fid}[1]{\ensuremath{\accentset{\triangledown}{#1}}}

\makeatletter
\@ifpackageloaded{tensor}
{

}{%

}
\makeatother

\makeatletter
\@ifpackageloaded{tensor}
{

}{%

}
\makeatother

\newcommand{\absnorm}[1]{\ensuremath{\left|#1\right|}}

\makeatletter
\@ifundefined{volume}{%
}%
{%
}
\makeatother

\newcommand{\tensortensor}[2]{\ensuremath{#1 \otimes #2}}
\makeatletter
\@ifpackageloaded{MnSymbol} 
{
\newcommand{\tensordot}[2]{\ensuremath{#1 \vdotdot #2}} 
}{%
\newcommand{\tensordot}[2]{\ensuremath{#1 : #2}} 
}
\makeatother

\newcommand{\vectordot}[2]{\ensuremath{#1 \bullet #2}}

\newcommand{\tensorschur}[2]{\ensuremath{#1 \circ #2}} 

\newcommand{\tensorf}[1]{{\mathfrak{#1}}}



%% file: vit-prusa-macros-reference-natural.tex
\newcommand{\nplacer}{\mathnormal{\kappa}_{\mathnormal{p}(\mathnormal{t})}}  

\newcommand{\fgradrng}{\ensuremath{\tensorq{G}}} 

\newcommand{\lcgnc}{\ensuremath{\lcg_{\nplacer}}} 

\newcommand{\fgradnc}{\ensuremath{\fgrad_{\nplacer}}} 

\newcommand{\detfgradnc}{\ensuremath{J_{\nplacer}}} 

\newcommand{\henckystrainnc}{\tensorq{H}_{\nplacer}} 

\newcommand{\henckync}{\henckystrainnc} 

\newcommand{\devhenckync}{\overline{\henckync}} 

\newcommand{\devhenckyncdelta}{\tensorq{H}_{\nplacer, \, \mathnormal{\delta}}} 

\newcommand{\sphenckync}{h_{\nplacer}} 

\newcommand{\thcstressnc}{\ensuremath{\cstress_{\nplacer}}}
\newcommand{\devthcstressncdelta}{\ensuremath{\cstress_{\nplacer, \, \mathnormal{\delta}}}} 
\newcommand{\thpressurenc}{\ensuremath{p_{\nplacer}}} 

\newcommand{\thcstressncrho}{\ensuremath{\tensorq{S}_{\nplacer}}} 
\newcommand{\thcstressncrhoc}{\ensuremath{\tensorc{S}_{\nplacer}}} 

\newcommand{\devthcstressncrho}{\ensuremath{\tensorq{S}_{\nplacer,\, \mathnormal{\delta}}}} 

\newcommand{\thpressurencrho}{\ensuremath{s_{\nplacer}}} 


%% file: vit-prusa-env-list.tex
\newtheorem{theorem}{Theorem} 

\theoremstyle{proposition}



%% file: summary-style.tex
\DeclareFloatingEnvironment[fileext=sum,placement={!ht},name=Summary]{summaryflt}

\newenvironment{summary}[1][]
    {
    \begin{summaryflt}[tb]
        \begin{summarycont}[#1] 
    }
    {
        \end{summarycont}
        \end{summaryflt}
    }

\mdfdefinestyle{summarystyle}{%
linecolor=black,linewidth=3pt,%
frametitlerule=true,%
frametitlebackgroundcolor=gray!20,
innertopmargin=\topskip,
}

\mdtheorem[style=summarystyle]{summarycont}{Summary}

\newmdenv[style=scholionstyle]{scholion}
\mdfdefinestyle{scholionstyle}{
  backgroundcolor=gray!10,
  roundcorner=10pt
}

%% file: text/implicit-fluids-body.tex
\section{Introduction}
\label{sec:introduction}
Most of the classical treatises on continuum mechanics when discussing constitutive theory start with the assumption that \emph{the stress is given as a function or a functional of the deformation}, see for example~\cite{truesdell.c.noll.w:non-linear*1}, \cite{muller.i:thermodynamics} and~\cite{coleman.bd.noll.w:approximation}. However, it has been argued that this approach is too restrictive, see especially~\cite{rajagopal.kr:on*3,rajagopal.kr:on*4}.

For example, it turns out that it is worthwhile to generalise the classical incompressible Navier--Stokes fluid model $\cstress = - p \identity + 2 \nu \gradsym$ as
\begin{subequations}
  \label{eq:1}
  \begin{align}
    \label{eq:2}
    \cstress &= - p \identity + \traceless{\cstress}, \\
    \label{eq:3}
    \gradsym  &= \tensorf{f} \left( \traceless{\cstress} \right),
  \end{align}
\end{subequations}
where we use the standard notation $\cstress$ for the Cauchy stress tensor, $p$ for the mechanical pressure, $\gradsym$ for the symmetric part of the velocity gradient and $\tensorf{f}$ for a tensorial function. This approach to the generalisation of the classical incompressible Navier--Stokes fluid model is \emph{in contrast with the classical approach} based on the formula
\begin{equation}
  \label{eq:143}
  \cstress = - p \identity + \tensorf{g} \left( \gradsym \right),  
\end{equation}
where $\tensorf{g}$ is a tensorial function. (Here one can think of classical models for shear-thinning/thickening fluids.) Note that even if the relation $\gradsym  = \tensorf{f} \left( \traceless{\cstress} \right)$ is invertible, and hence can be rewritten in the classical form $\traceless{\cstress} = \tensorf{g} \left( \gradsym \right)$, it might be still convenient to use the representation~\eqref{eq:3} because the representation~\eqref{eq:3} might be much simpler than the classical one. (Meaning that the formula $\gradsym  = \tensorf{f} \left( \traceless{\cstress} \right)$ might have a nice analytical form, while the inverse formula $\traceless{\cstress} = \tensorf{g} \left( \gradsym \right)$ could be complicated or impossible to write down using elementary functions.) If~\eqref{eq:3} is not invertible, then the two approaches are clearly strikingly different. Furthermore one can even go one step further and replace~\eqref{eq:3} by a more general \emph{implicit relation} $\tensorf{k} \left( \traceless{\cstress}, \gradsym \right) = \tensorzero$, where $\tensorf{k}$ is a tensorial function and~$\tensorzero$ denotes the zero tensor.
 
Similarly, in the context of mathematical models for elastic solids, one can generalise the standard isotropic Cauchy elastic materials wherein one assumes that $\cstress = \tensorf{h}\left( \lcg \right)$, where $\tensorf{h}$ is a tensorial function, and $\lcg$ denotes the left Cauchy--Green tensor. The generalisation of this formula is
\begin{equation}
  \label{eq:4}
  \lcg = \tensorf{i}\left( \cstress \right)
\end{equation}
or even $\tensorf{j}\left( \cstress, \lcg \right)  = \tensorzero$, where $\tensorf{i}$ and $\tensorf{j}$ are again tensorial functions. Clearly, everything that has been said in the context of models for fluids holds as well for the models for elastic solids. 

\emph{Since viscoelastic fluids are non-Newtonian fluids that exhibit both ``viscous'' and ``elastic'' characteristics, it is desirable to see how to exploit the concept of implicit type constitutive relations in the context of mathematical models for these fluids}. In particular, the phenomenological theory of viscoelastic rate-type fluids is based on the assumption that the elastic response of the fluid is modelled using the same concepts as in the classical theory of nonlinear solids. (Especially the assumed formulae for the Helmholtz free energy borrow heavily from the classical theory of nonlinear solids.) However, if one wants to model the elastic response in the spirit of recent advances in theory of nonlinear elastic solids that exploit constitutive relations of type~\eqref{eq:4}, the phenomenological derivation of mathematical models viscoelastic rate-type fluids must be modified accordingly. Such modifications are discussed---with a special emphasis on the corresponding thermodynamic basis---in the current contribution.

The paper is organised as follows. First we briefly review the concept of constitutive relations for fluids and solids, that are developed using the approach~\eqref{eq:1} or~\eqref{eq:4} and its generalisations to the fully implicit setting $\tensorf{k} \left( \traceless{\cstress}, \gradsym \right) = \tensorzero$ or $\tensorf{j}\left( \cstress, \lcg \right)  = \tensorzero$, see Section~\ref{sec:mater-spec-via}. Then we proceed with a brief review of the classical derivation of viscoelastic rate-type models that is based on the concept of natural configuration and the characterisation of the elastic response with the Helmholtz free energy, see Section~\ref{sec:visc-rate-type-1}.

In Section~\ref{sec:gibbs-free-energy} we identify the Hencky strain tensor as the convenient tensorial quantity and the Gibbs free energy as the convenient thermodynamic potential that allows one to work with (elastic) constitutive response of type~\eqref{eq:4}, and we outline a general procedure that allows one to develop novel viscoelastic rate-type models based on these concepts. In Section~\ref{sec:gies-visc-rate} we study the classical Giesekus/Oldroyd-B model both in the classical framework and the newly proposed framework. For this model it is possible to write down explicitly all formulae in both approaches, which allows us to clearly document the general theory. Finally, we also discuss the applicability of the maximisation of the entropy production hypothesis, see~\cite{rajagopal.kr.srinivasa.ar:on*7}, in the theory of constitutive relations.

\section{Materials specified via implicit constitutive relations}
\label{sec:mater-spec-via}

Before we proceed with viscoelastic fluids, we briefly summarise state-of-the-art regarding the novel approach to the mathematical modelling of \emph{viscous} fluids and \emph{elastic} solids.

\subsection{Fluids}
\label{sec:fluids}
A simple generalisation of the Navier--Stokes fluid as indicated in~\eqref{eq:1} is the constitutive relation in the form
\begin{equation}
  \label{eq:5}
  \gradsym
  =
  \left(
    \alpha
    \left(
      1
      +
      \beta
      \absnorm{\traceless{\cstress}}^2
    \right)^n
    +
    \delta
  \right)
  \traceless{\cstress}
  ,
\end{equation}
introduced in \cite{roux.c.rajagopal.kr:shear}. (For a through discussion of a simpler model with $\delta=0$ see~\cite{malek.j.prusa.ea:generalizations}.) With properly tuned parameter values this constitutive relation leads to the characteristic S-shaped curve in the strain rate--stress diagram, hence the model~\eqref{eq:5} can serve as simple phenomenological model for flows of substances that exhibit such behaviour, see for example \cite{boltenhagen.p.hu.y.ea:observation}, \cite{grob.m.heussinger.c.ea:jamming} and \cite{mari.r.seto.r.ea:nonmonotonic} to name a few. (Further references and a thorough discussion can be found in \cite{perlacova.t.prusa.v:tensorial} and \cite{janecka.a.prusa.v:perspectives}.) Note that if the constitutive relation leads to the S-shaped curve in the strain rate--stress diagram, then~\eqref{eq:5} is not invertible, and it can not be rewritten in the classical form $\traceless{\cstress} = \tensorf{g} \left( \gradsym \right)$. Other models that are easy to describe using~\eqref{eq:1} are the models for incompressible fluids with pressure dependent viscosity, see~\cite{rajagopal.kr:on*4}, or for that matter models wherein the viscosity depends on the stress such as the models used for the flow of ice, see~\cite{blatter.h:velocity} and \cite{pettit.ec.waddington.ed:ice} to name a few and also~\cite{malek.j.prusa.v:derivation} for further references regarding this class of models.

Other models that are easy to decscribe using~\eqref{eq:1} or the implicit relation $\tensorf{k} \left( \traceless{\cstress}, \gradsym \right) = \tensorzero$ are the models with activation criterion introduced by~\cite{bingham.ce:fluidity} and \cite{herschel.wh.bulkley.r:konsistenzmessungen}. These models are usually, see for example~\cite{duvaut.g.lions.jl:inequalities}, written down in the form of a dichotomy relation 
\begin{equation}
\label{eq:6}
\absnorm{\traceless{\cstress}} \le \tau_{*} \Leftrightarrow \gradsym= \tensorzero \quad\text{and}\quad
\absnorm{\traceless{\cstress} }>\tau_{*} \Leftrightarrow \traceless{\cstress}=\frac{\tau_{*} \gradsym}{\absnorm{\gradsym}} +
2\nu (\absnorm{\gradsym}) \gradsym,
\end{equation}
where $\nu$ is a positive function and $\tau_*$ is the yield stress.) If we use the alternative framework based on the constitutive relation~$\tensorf{k} \left( \traceless{\cstress}, \gradsym \right) = \tensorzero$, then these models can be rewritten as 
\begin{equation}
  \label{eq:7}
\gradsym = \frac{1}{2 \nu (\absnorm{\gradsym})} \frac{\left(\absnorm{\traceless{\cstress}} - \tau_*\right)^{+}}{\absnorm{\traceless{\cstress}}} \traceless{\cstress},
\end{equation}
where $x^{+}=\max\{x,0\}$ and $\nu$ is a non-negative function. This observation can be exploited in the discussion of the physical and mathematical properties of the models, see for example~\cite{rajagopal.kr.srinivasa.ar:on*8} and \cite{bulcek.m.gwiazda.p.ea:on*3}. For further discussion of generalised models of type $\tensorf{k} \left( \traceless{\cstress}, \gradsym \right) = \tensorzero$ we also refer the reader to~\cite{perlacova.t.prusa.v:tensorial}. Note that in all cases discussed above the constitutive relation between the stress and the symmetric part of the velocity gradient has been related to entropy production mechanisms, see for example~\cite{janecka.a.prusa.v:perspectives} for a through discussion.

Reader interested in analytical and semianalytical solutions to boundary value problems for fluids described by constitutive relations of the type  $\traceless{\cstress} = \tensorf{g} \left( \gradsym \right)$ or $\tensorf{k} \left( \traceless{\cstress}, \gradsym \right) = \tensorzero$ is referred to~\cite{malek.j.prusa.ea:generalizations}, \cite{roux.c.rajagopal.kr:shear}, \cite{srinivasan.s.karra.s:flow}, \cite{narayan.spa.rajagopal.kr:unsteady}, \cite{fusi.l.farina.a:flow}, \cite{fusi.l:channel}, \cite{housiadas.kd.georgiou.gc.ea:note}, \cite{gomez-constante.jp.rajagopal.kr:flow} and~\cite{fetecau.c.bridges.c:analytical} to name a few. Numerical solution of the corresponding governing equations is investigated in~\cite{janecka.a.malek.j.ea:numerical}, while rigorous numerical analysis for various models that fall into this class is discussed in~\cite{diening.l.kreuzer.c.ea:finite}, \cite{stebel.j:finite,hirn.a.lanzendorfer.m.ea:finite}, \cite{suli.e.tscherpel.t:fully}, \cite{farrell.pe.gazca-orozco.pa.ea:numerical} and \cite{farrell.pe.gazca-orozco.pa:augmented}. Rigorous mathematical theory for some of the aforementioned models is developed in \cite{bulcek.m.gwiazda.p.ea:on*2,bulcek.m.gwiazda.p.ea:on*3} and~\cite{maringova.e.zabensky.j:on}, see also~\cite{blechta.j.malek.j.ea:on} for a recent review.

Further generalisation of the algebraic relation $\tensorf{k} \left( \traceless{\cstress}, \gradsym \right) = \tensorzero$ is the implicit relation of the type
\begin{equation}
\label{eq:8}
    \fhistory{H}{s=0}{+\infty}\left(\cstress(t-s), \rcgrel{t}{t-s}\right)  = \tensorzero,
\end{equation}
that relates the histories of the stress tensor and the relative right Cauchy--Green tensor, for details see~\cite{prusa.rajagopal.kr:on}. Special instances of~\eqref{eq:8} are rate-type and differential-type models for viscoelastic fluids and integral models for viscoelastic fluids, while in all these cases one can consider models with pressure/stress dependent material parameters, see for example~\cite{karra.s.prusa.ea:on}, \cite{kannan.k.rajagopal.kr:model}, \cite{housiadas.kd:internal,housiadas.kd:viscoelastic} and~\cite{arcangioli.m.farina.a.ea:constitutive} and references therein.

\subsection{Solids}
\label{sec:solids}
Regarding the similar development in the case of constitutive relations for solids, wherein the classical approach is based on constitutive relations in the form $\cstress = \tensorf{h}\left( \lcg \right)$. The non-standard relations of the type $\lcg = \tensorf{i}\left( \cstress \right)$ or even $\tensorf{j}\left( \cstress, \lcg \right)  = \tensorzero$ have been considered without a proper thermodynamic basis or clearly articulated rationale for the choice of such constitutive relations in some early works such as~\cite{morgan.aja:some}, \cite{fitzgerald.je:tensorial} or~\cite{blume.ja:on} and \cite{xiao.h.chen.l:hencky,xiao.h.bruhns.ot.ea:explicit}. A systematic study of constitutive relations of the type $\lcg = \tensorf{i}\left( \cstress \right)$ with coherent explanation based on issues of causality for the choice of such constitutive relations was initiated by~\cite{rajagopal.kr:on*3}, and the number of works focused on this concept has been growing ever since. Regarding elastic (non-dissipative) solids we refer the reader to a recent review by~\cite{bustamante.r.rajagopal.kr:review} and references therein, while some aspects of the description of non-elastic response such as plasticity are discussed in~\cite{rajagopal.kr.srinivasa.ar:implicit*1,rajagopal.kr.srinivasa.ar:inelastic} and~\cite{cichra.d.prusa.v:thermodynamic}. (A discussion and references on works dealing the invertibility of the constitutive relations can be found in~\cite{sfyris.d.bustamante.r:use}, one of the earliest studies of the same is that by~\cite{truesdell.c.moon.h:inequalities}.) As an example of the benefits of the proposed change of perspective we can mention especially the study by~\cite{muliana.a.rajagopal.kr.ea:determining}, who address the classical topic of mathematical modelling of the mechanical response of rubber.

From the current perspective, the important aspects of the novel approach to the constitutive theory of solids are the following. First, if one wants to work with constitutive relations of the type~$\lcg = \tensorf{i}\left( \cstress \right)$, then it is convenient to work with the Hencky strain tensor $\henckystrain =_{\bydefinition} \frac{1}{2} \ln \lcg$ instead of the left Cauchy--Green tensor $\lcg$. (This observation is made in various works in elasticity and inelasticity, see especially~\cite{xiao.h.bruhns.ot.ea:explicit,xiao.h.chen.l:hencky,xiao.h.bruhns.ot.ea:elastoplasticity}. Note also that in the context of viscoelastic fluids the preferred quantity to work with is---from the numerical perspective---the log-conformation tensor, see~\cite{fattal.r.kupferman.r:constitutive,fattal.r.kupferman.r:time-dependent,hulsen.ma.fattal.r.ea:flow}, that is the logarithm of the conformation tensor. Apparently, this is not a coincidence.) Second, the thermodynamic potential of choice is the Gibbs free energy, see~\cite{rajagopal.kr.srinivasa.ar:gibbs-potential-based} and~\cite{srinivasa.ar:on}. In particular, in the case of elastic solids the formula for the Hencky strain $\henckystrain = \tensorf{i}\left( \cstress \right)$ arises via the differentiation of a potential, and the potential can be identified as the Gibbs free energy, see especially~\cite{gokulnath.c.saravanan.u.ea:representations} and~\cite{prusa.v.rajagopal.kr.ea:gibbs}. 

\section{Viscoelastic rate-type fluids}
\label{sec:visc-rate-type-1}
Since viscoelastic fluids are fluids that exhibit both ``viscous'' and ``elastic'' characteristics, see for example~\cite{snoeijer.jh.pandey.a.ea:relationship} for a recent discussion concerning the same, one can ask whether the novel approach to the modelling of the response elastic solids can be also applied in this context. The point is the following. Traditionally the elastic properties of viscoelastic fluids are from the phenomenological perspective modelled using the same concepts as in the classical theory of nonlinear solids wherein the Cauchy stress tensor is assumed to be a function of the left Cauchy--Green tensor. \emph{The question is what happens if an opposite perspective to the classical one is taken}---the left Cauchy--Green stress tensor is assumed to be a function the stress tensor, or one even assumes that these quantities are related by a implicit algebraic relation.

We focus on this question, and we show how to use the new perspective in the development of mathematical models for viscoelastic fluids, and discuss in detail the thermodynamic underpinnings of such models. In particular we focus on the use of Gibbs free energy instead of the Helmholtz free energy. Using the standard Giesekus/Oldroyd-B models, we show how the alternative approach works in the case of well-known models, and we argue that the proposed approach is straightforward to generalise to more complex settings.

Our basic framework will be the framework based on the concept of evolving natural configuration, see~\cite{eckart.c:thermodynamics*3} and~\cite{rajagopal.kr.srinivasa.ar:thermodynamic}, that has been successfully applied in various settings, see for example~\cite{malek.j.rajagopal.kr.ea:thermodynamically,narayan.spa.little.dn.ea:modelling,malek.j.rajagopal.kr.ea:thermodynamically*1}, \cite{malek.j.rajagopal.kr.ea:derivation}, \cite{tuma.k.stein.j.ea:motion}, \cite{malek.j.prusa.v.ea:thermodynamics}, \cite{rehor.m.gansen.a.ea:comparison} or \cite{sumith.s.kannan.k.ea:constitutive} to name a few. (For other approaches to the mathematical modelling of viscoelastic rate-type fluids see~\cite{leonov.ai:nonequilibrium,leonov.ai:on}, \cite{wapperom.p.hulsen.ma:thermodynamics} or the GENERIC framework, see~\cite{grmela.m.ottinger.hc:dynamics,ottinger.hc.grmela.m:dynamics} and \cite{pavelka.m.klika.v.ea:multiscale}.)

\section{Gibbs free energy based approach to viscoelastic rate-type fluids}
\label{sec:gibbs-free-energy}

The concept of evolving natural configuration is based on the assumption that the overall response of a viscoelastic fluid is composed of a dissipative (viscous) response and a non-dissipative (elastic) response, while the latter can be described using the concepts known from the theory of nonlinear elasticity. However, the left Cauchy--Green tensor for the total response $\lcg$ must be replaced by the left Cauchy--Green tensor $\lcgnc$ associated to the elastic part of the fluid response. The subscript $\nplacer$ denotes the instantaneously relaxed or ``natural'' configuration, see~\cite{rajagopal.kr.srinivasa.ar:thermodynamic} for a detailed rationale. (Another detailed explanation of the concept of natural configuration is given in a recent review article by~\cite{malek.j.prusa.v:derivation}. However, note that in the current contribution we do not explicitly use the decomposition of the total deformation gradient $\fgrad = \fgradnc \fgradrng$ to the dissipative and non-dissipative part. In particular we do not explicitly work with the tensor $\fgradrng$, which makes our approach different from the treatment in~\cite{rajagopal.kr.srinivasa.ar:thermodynamic} and~\cite{malek.j.rajagopal.kr.ea:on}.) If we want to exploit the novel approach to the modelling of elastic response, we have to mimic~\eqref{eq:4}, but the full Cauchy--Green tensor $\lcg$ must be replaced by the partial Cauchy--Green tensor~$\lcgnc$. 

Furthermore, as we have already indicated, the partial left Cauchy--Green tensor $\lcgnc$ should be replaced by the corresponding Hencky strain tensor, and the thermodynamic potential of choice should be the Gibbs free energy. The details are worked out below. Note that unlike in other works focused on the use of Gibbs potential, see especially~\cite{narayan.spa.little.dn.ea:nonlinear}, we are not assuming that the processes of interest are isothermal. The theory outlined below is applicable in non-isothermal setting.

The main idea behind the presented approach is that the behaviour of the material in the processes of interest is determined by two factors, namely its ability to \emph{store energy} and \emph{produce entropy}. The energy storage mechanisms are in the present case specified by the choice of the Gibbs free energy $\gibbs$, while the entropy production mechanisms are specified by the choice of the formula for the entropy production $\entprodc$. Since the entropy $\entropy$ is obtained via differentiation of $\gibbs$, the evolution equation for $\entropy$ then necessarily links the assumed form of the Gibbs free energy $\gibbs$ and the assumed entropy production $\entprodc$. This link can be exploited in the identification of the constitutive relations. Consequently, the first step of the proposed approach is to derive an evolution equation for the entropy.

\subsection{Entropy evolution equation -- general case}
\label{sec:entr-evol-equat}

The fundamental equation in continuum mechanics is the evolution equation for the specific internal energy $\ienergy$, that is for the internal energy per unit mass, $\left[ \ienergy \right] = \unitfrac{J}{kg}$ reads
\begin{equation}
  \label{eq:9}
  \rho
  \dd{\ienergy}{t} = \tensordot{\cstress}{\gradsym} - \divergence \hfluxc,
\end{equation}
see standard texts on continuum thermodynamics such as \cite{truesdell.c.noll.w:non-linear*1} or \cite{muller.i:thermodynamics}. (The notation in~(\ref{eq:9}) is the standard one, the symbol $\cstress$ stands for the Cauchy stress tensor, $\gradsym$ denotes the symmetric part of the velocity gradient, $\rho$ denotes the density, and $\hfluxc$ denotes the heat flux.) Evolution equation~\eqref{eq:9} is the starting point for the derivation of the sought entropy evolution equation.

In the present approach to the theory of viscoelastic rate-type fluids, we assume that the specific internal energy is a function of the specific entropy $\entropy$, density $\rho$, and the Hencky strain tensor associated to the elastic response, 
\begin{equation}
  \label{eq:10}
  \henckync =_{\bydefinition} \frac{1}{2} \ln \lcgnc.
\end{equation}
Furthermore, it is convenient to split the Hencky strain tensor to the traceless (deviatoric) part $\devhenckyncdelta$ and the spherical part $\sphenckync$,
\begin{equation}
  \label{eq:11}
  \henckync = \sphenckync \identity + \devhenckyncdelta,
\end{equation}
where we introduce the notation $\traceless{\generictensor} = _{\bydefinition} \generictensor - \frac{1}{3} \left(\Tr \generictensor \right) \identity$ for the traceless (deviatoric) part of the corresponding tensor. We note that the deviatoric part of the Hencky strain tensor $\devhenckyncdelta$ is equal to the rescaled Hencky strain tensor 
\begin{equation}
  \label{eq:12}
  \devhenckync =_{\bydefinition} \frac{1}{2} \ln \overline{\lcgnc}
\end{equation}
associated to the rescaled left Cauchy--Green tensor $\overline{\lcgnc}$ that is defined as
\begin{equation}
  \label{eq:13}
  \overline{\lcgnc} =_{\bydefinition} \frac{\lcgnc}{\detfgradnc^{\frac{2}{3}}},
\end{equation}
where $\detfgradnc =_{\bydefinition} \det \fgradnc$. The equality $\devhenckyncdelta = \devhenckync$ is discussed for example in~\cite{prusa.v.rajagopal.kr.ea:gibbs}, see also~\cite{xiao.h.bruhns.ot.ea:explicit}. The usage of the rescaled left Cauchy--Green tensor~$\overline{\lcgnc}$ is a common practice in the theory of slightly compressible solids, see for example~\cite{horgan.co.saccomandi.g:constitutive}. The rationale is that $\det \overline{\lcgnc} = 1$, hence the deformation is conveniently split into a volume-preserving and volume-changing part. We also note that if the elastic part of the response is volume-preserving, then $\det \lcgnc = 1$ and consequently $\Tr \henckync = 0$, which in this special case implies that $\henckync = \devhenckyncdelta = \devhenckync$.

Regarding the specific internal energy $\ienergy$ we therefore assume that
\begin{equation}
  \label{eq:14}
  \ienergy = \ienergy \left( \entropy, \rho, \sphenckync, \devhenckyncdelta \right),
\end{equation}
which allows us to seamlessly deal with incompressible materials as well. (If the material is homogeneous and incompressible, we can simply remove $\rho$ from~\eqref{eq:14} since the density is in this case constant. Similarly, if the response of the elastic component is modelled as a response of an incompressible material, then $\sphenckync$ is a constant and we can simply remove $\sphenckync$ form~\eqref{eq:14}. We however need to introduce the corresponding Lagrange multipliers that enforce the incompressibility constraint.)
Introducing the thermodynamic temperature $\temp$ in the standard manner as
\begin{equation}
  \label{eq:15}
  \temp =_{\bydefinition} \pd{\ienergy}{\entropy}(\entropy, \rho, \sphenckync, \devhenckyncdelta),
\end{equation}
we can rewrite~\eqref{eq:9} in terms of the Helmholtz free energy $\fenergy$,
\begin{equation}
  \label{eq:16}
  \fenergy (\temp, \rho, \sphenckync, \devhenckyncdelta)
  =_{\bydefinition}
  \left.
    \left(
      \ienergy(\entropy, \rho, \sphenckync, \devhenckyncdelta) - \temp \entropy
    \right)
  \right|_{\entropy = \entropy(\temp,\, \rho, \, \sphenckync, \, \devhenckyncdelta)},
\end{equation}
as
\begin{multline}
  \label{eq:17}
  \rho \temp \dd{\entropy}{t}
  =
  \tensordot{\cstress}{\gradsym}
  -
  \rho
  \pd{\fenergy}{\rho}(\temp, \rho, \sphenckync, \devhenckyncdelta)\dd{\rho}{t}
  -
  \rho
  \pd{\fenergy}{\sphenckync}(\temp, \rho, \sphenckync, \devhenckyncdelta)\dd{\sphenckync}{t}
  \\
  -
  \rho
  \tensordot{\traceless{\left(  \pd{\fenergy}{\devhenckyncdelta}(\temp, \rho, \sphenckync, \devhenckyncdelta) \right)}}{\dd{\devhenckyncdelta}{t}}
  -
  \divergence \hfluxc
  .
\end{multline}
(In the next to the last term we can consider only the traceless part of $\pd{\fenergy}{\devhenckyncdelta}$ since it enters the dot product with the \emph{traceless} tensor $\dd{\devhenckyncdelta}{t}$.) Finally we can split the Cauchy stress tensor into the mean normal stress $\mns=_{\bydefinition} \frac{1}{3} \Tr \cstress$ and the deviatoric part $\dcstresssymb$,
\begin{equation}
  \label{eq:18}
  \cstress = \mns \identity + \dcstresssymb, 
\end{equation}
and we get
\begin{multline}
  \label{eq:19}
  \rho \temp \dd{\entropy}{t}
  =
  \mns \divergence \vec{v}
  +
  \tensordot{\dcstresssymb}{\traceless{\gradsym}}
  +
  \rho^2
  \pd{\fenergy}{\rho}(\temp, \rho, \sphenckync, \devhenckyncdelta)\divergence \vec{v}
  \\
  -
  \rho
  \pd{\fenergy}{\sphenckync}(\temp, \rho, \sphenckync, \devhenckyncdelta)\dd{\sphenckync}{t}
  -
  \rho
  \tensordot{\traceless{\left(\pd{\fenergy}{\devhenckyncdelta}(\temp, \rho, \sphenckync, \devhenckyncdelta)\right)}}{\dd{\devhenckyncdelta}{t}}
  -
  \divergence \hfluxc
  .
\end{multline}
(We have also used the balance of mass $\dd{\rho}{t} + \rho \divergence \vec{v} = 0$.) We recall that the specific internal energy and the specific Helmholtz free energy are related by the Legendre transformation, and that the following formula holds
\begin{equation}
  \label{eq:20}
  \entropy = - \pd{\fenergy}{\temp}(\temp, \rho, \sphenckync, \devhenckyncdelta).
\end{equation}

The manipulation outlined above therefore gives us the sought evolution equation~\eqref{eq:17} for the entropy~$\entropy$. \emph{This equation is the starting point for the specification of the constitutive relations}. In the standard approach, one in principle wants to provide a constitutive relation for the Cauchy stress tensor~$\cstress$ and formulae for the time derivatives of $\sphenckync$ and $\devhenckyncdelta$ such that the right-hand side of~\eqref{eq:17} is nonnegative, which in turn guarantees that the second law of thermodynamics holds in the given material and during the given class of processes.  (See for example~\cite{dostalk.m.prusa.v.ea:on} for applications of this procedure to several well known viscoelastic rate-type models.) We however might not want to use~\eqref{eq:17} directly, we want to rewrite it in terms of Gibbs free energy $\gibbs$. 

In order to use the Gibbs free energy we need to introduce the thermodynamic pressure $\thpressure$, and the stress $\thcstressnc$, which we again split into the spherical part $\thpressurenc$ and the traceless part~$\devthcstressncdelta$,
\begin{equation}
  \label{eq:21}
  \thcstressnc =_{\bydefinition} \thpressurenc \identity + \devthcstressncdelta. 
\end{equation}
Using the definitions
\begin{subequations}
  \label{eq:22}
  \begin{align}
    \label{eq:23}
    \thpressure &=_{\bydefinition} \rho^2 \pd{\fenergy}{\rho}(\temp, \rho, \sphenckync, \devhenckyncdelta), \\
    \label{eq:24}
    \frac{\thpressurenc}{\rho} &=_{\bydefinition} \pd{\fenergy}{\sphenckync}(\temp, \rho, \sphenckync, \devhenckyncdelta)
    \\
    \label{eq:25}
    \frac{\devthcstressncdelta}{\rho}
                &=_{\bydefinition}
                  \traceless{ \left( \pd{\fenergy}{\devhenckyncdelta}(\temp, \rho, \sphenckync, \devhenckyncdelta) \right)},
  \end{align}
\end{subequations}
we can further rewrite~\eqref{eq:17} in such a way that it is ready for the use of the specific Gibbs free energy $\gibbs$. (The definition~\eqref{eq:23} is the standard definition of the thermodynamic pressure as it is used in the classical thermodynamics of compressible fluids.) If we further introduce the notation
\begin{subequations}
  \label{eq:26}
  \begin{align}
    \thpressurencrho
    &=_{\bydefinition}
      \frac{\thpressurenc}{\rho}
    \\
    \label{eq:27}
    \devthcstressncrho
    &=_{\bydefinition}
      \frac{\devthcstressncdelta}{\rho},
  \end{align}
\end{subequations}
the specific Gibbs free energy is then defined as
\begin{multline}
  \label{eq:28}
  \gibbs \left( \temp, \thpressure, \thpressurencrho, \devthcstressncrho \right)
  \\
  =_{\bydefinition}
  \left.
    \left(
      \fenergy(\temp, \rho, \sphenckync, \devhenckyncdelta)
      +
      \frac{\thpressure}{\rho}
      -
      \thpressurencrho \sphenckync
      -
      \tensordot{\devthcstressncrho}{\devhenckyncdelta}
    \right)
  \right|_{
    \rho = \rho \left( \temp,\, \thpressure,\, \thpressurencrho,\, \devthcstressncrho  \right),\, \cdots
  }
  .
\end{multline}

We note that the transition to the Gibbs free energy is in fact done by another Legendre transformation, and that the use of Legendre transformation implies that
\begin{subequations}
  \label{eq:29}
  \begin{align}
    \label{eq:30}
    \entropy &= - \pd{\gibbs}{\temp} \left( \temp, \thpressure, \thpressurencrho, \devthcstressncrho \right), \\
    \label{eq:31}
    \frac{1}{\rho} &= \pd{\gibbs}{\thpressure} \left( \temp, \thpressure, \thpressurencrho, \devthcstressncrho \right), \\
    \label{eq:32}
    \sphenckync &= - \pd{\gibbs}{\thpressurencrho} \left( \temp, \thpressure, \thpressurencrho, \devthcstressncrho \right).
    \\
    \label{eq:33}
    \devhenckyncdelta &= - \traceless{\left( \pd{\gibbs}{\devthcstressncrho} \left( \temp, \thpressure, \thpressurencrho, \devthcstressncrho \right) \right)}.
  \end{align}
\end{subequations}
If we use definitions~\eqref{eq:22}, we see that~\eqref{eq:19} can be rewritten as
\begin{equation}
  \label{eq:34}
  \rho \temp \dd{\entropy}{t}
  =
  -
  \frac{\mns + \thpressure}{\rho} \dd{\rho}{t}
  +
  \tensordot{\dcstresssymb}{\traceless{\gradsym}}
  -
  \thpressurenc
  \dd{\sphenckync}{t}
  -
  \tensordot{\devthcstressncdelta}{\dd{\devhenckyncdelta}{t}}
  -
  \divergence \hfluxc
  ,
\end{equation}
which upon using formulae~\eqref{eq:29} yields
\begin{multline}
  \label{eq:35}
  \rho \temp \dd{\entropy}{t}
  =
  -
  \rho
  \left(\mns + \thpressure \right)
  \dd{}{t}
  \left(
    \pd{\gibbs}{\thpressure} \left( \temp, \thpressure, \thpressurencrho, \devthcstressncrho \right)
  \right)
  +
  \tensordot{\dcstresssymb}{\traceless{\gradsym}}
  +
  \thpressurenc
  \dd{}{t}
  \left(
    \pd{\gibbs}{\thpressurencrho} \left( \temp, \thpressure, \thpressurencrho, \devthcstressncrho \right)
  \right)
  \\
  +
  \tensordot{\devthcstressncdelta}
  {
    \dd{}{t}
    \left(
      \traceless{
        \left(
          \pd{\gibbs}{\devthcstressncrho} \left( \temp, \thpressure, \thpressurencrho, \devthcstressncrho \right)
        \right)
      }
    \right)
  }
  -
  \divergence \hfluxc
  .
\end{multline}
This is a general evolution equation for the specific entropy $\entropy$.

\subsection{Entropy evolution equation -- incompressible fluids}
\label{sec:entr-evol-equat-1}

Let us now manipulate~\eqref{eq:34} into a form convenient for the ongoing investigation of the constitutive relations. \emph{For the sake of simplicity we now restrict ourselves to incompressible fluids, which means that the density is a constant and it no longer enters the formula for the Gibbs free energy. Further we will not split the elastic response into the volume-preserving and volume-changing part, that is we use the full Hencky strain tensor $\henckync$ instead of the pair $\devhenckyncdelta$ and $\sphenckync$, and similarly for the reduced stress $\thcstressncrho$.} In this case the Gibbs free energy is therefore given as
\begin{equation}
  \label{eq:36}
  \gibbs = \gibbs \left( \temp, \thcstressncrho \right),
\end{equation}
and the evolution equation for the entropy~\eqref{eq:35} reduces to
\begin{equation}
  \label{eq:37}
  \rho \temp \dd{\entropy}{t}
  =
  \tensordot{\dcstresssymb}{\traceless{\gradsym}}
  \\
  +
  \tensordot{\thcstressnc}
  {
    \dd{}{t}
    \left[
      \pd{\gibbs}{\thcstressncrho} \left( \temp, \thcstressncrho \right)
    \right]
  }
  -
  \divergence \hfluxc
  ,
\end{equation}
and relations~\eqref{eq:29} read
\begin{subequations}
  \begin{align}
    \label{eq:38}
    \entropy &= - \pd{\gibbs}{\temp} \left( \temp, \thcstressncrho \right), \\
    \label{eq:39}
    \henckync &= - \pd{\gibbs}{\thcstressncrho} \left( \temp, \thcstressncrho \right).
  \end{align}
\end{subequations}

We note that if we deal with an incompressible substance, then the Cauchy stress tensor is split to the spherical part $-p \identity$ and the deviatoric part $\traceless{\cstress}$,
\begin{equation}
  \label{eq:40}
  \cstress = -p\identity + \traceless{\cstress},
\end{equation}
and that the pressure $p$---which can be now understood as the Lagrange multiplier enforcing the incompressibility constraint---constitutes a new unknown field to be solved for. Namely, it is not given by an equation of state as a function of density, temperature and other variables. Furthermore the incompressibility constraint $\divergence \vec{v} = 0$ implies that $\traceless \gradsym = \gradsym$.

Finally, we note that if we deal with isotropic elastic response, which is the case in the rest of the paper, then~\eqref{eq:39} in virtue of the standard representation theorem for isotropic functions, see for example~\cite{truesdell.c.noll.w:non-linear*1}, implies that the tensors $\henckync$ and $\thcstressncrho$ commute, that is $\henckync\thcstressncrho = \thcstressncrho\henckync$.

\subsubsection{Manipulations with entropy production -- direct use of Gibbs free energy}
\label{sec:manip-with-entr}

If we further \emph{assume that the Gibbs free energy has the property that the mixed derivative vanishes}, that is if
\begin{equation}
  \label{eq:41}
  \pd{^2\gibbs}{\temp \partial \thcstressncrho} = 0,
\end{equation}
then~\eqref{eq:35} reduces to\footnote{The notation might be slightly ambiguous here. Using the index notation we have
  \begin{equation*}
    \tensordot{\thcstressncrho}{ \ppd{\gibbs}{\thcstressncrho} \left( \temp, \thcstressncrho \right) \dd{\thcstressncrho}{t}}
    =
    \tensor{{\thcstressncrhoc}}{_{ij}}
    \pd{^2\gibbs}{\tensor{{\thcstressncrhoc}}{_{ij}} \partial \tensor{{\thcstressncrhoc}}{_{mn}}}
    \dd{\tensor{{\thcstressncrhoc}}{_{mn}}}{t}
    ,
  \end{equation*}
  where we have used the summation convention.
}%
\begin{equation}
  \label{eq:42}
  \rho \temp \dd{\entropy}{t}
  =
  \tensordot{\traceless{\cstress}}{\traceless{\gradsym}}
  +
  \rho
  \tensordot{\thcstressncrho}{ \ppd{\gibbs}{\thcstressncrho} \left( \temp, \thcstressncrho \right) \dd{\thcstressncrho}{t}}
  -
  \divergence \hfluxc
  .
\end{equation}
The reduced stress tensor $\thcstressncrho$ is a second order tensor in the current configuration that transforms \emph{normals} to infinitesimal surfaces to traction \emph{vectors}, hence we see that it is a tensor of type $\binom{2}{0}$. If we need to calculate the time derivative of such a tensor, then the natural concept is the upper convected derivative,
\begin{equation}
  \label{eq:43}
  \fid{\overline{\thcstressncrho}}
  =_{\bydefinition}
  \dd{\thcstressncrho}{t}
  -
  \gradvl
  \thcstressncrho
  -
  \thcstressncrho
  \transpose{\gradvl}
  .
\end{equation}
(See for example~\cite{stumpf.h.hoppe.u:application} for geometrical underpinnings of the concept of upper convected derivative and its link to the concept of Lie derivative.) Using the definition of the upper convected derivative, we see that~\eqref{eq:42} can be rewritten as\footnote{For the sake of clarity we can resort to the index notation. We set
  \begin{equation*}
    \tensor{
      \left[
        \thcstressncrho
        \ppd{\gibbs}{\thcstressncrho}
        \thcstressncrho
      \right]
    }{_{mp}}
    =
    \tensor{{\thcstressncrhoc}}{_{ij}}
    \pd{^2\gibbs}{\tensor{{\thcstressncrhoc}}{_{ij}} \partial\tensor{{\thcstressncrhoc}}{_{mn}}}
    \tensor{{\thcstressncrhoc}}{_{np}}
    ,
  \end{equation*}
  where we use the summation convention. Since $\thcstressncrho$ is a symmetric tensor that commutes with $\pd{\gibbs}{\thcstressncrho}$, we see that
  $
  \thcstressncrho
  \ppd{\gibbs}{\thcstressncrho}
  \thcstressncrho
  $
  is also a symmetric tensor.
}
\begin{equation}
  \label{eq:44}
  \rho \temp \dd{\entropy}{t}
  =
  \tensordot{
    \left\{
      \traceless{\cstress}
      +
      2
      \rho
      \traceless{
        \left[
          \thcstressncrho
          \ppd{\gibbs}{\thcstressncrho}
          \left(
            \temp, \thcstressncrho
          \right)
          \thcstressncrho
        \right]
      }
    \right\}
  }
  {\traceless{\gradsym}}
  +
  \rho
  \tensordot{
    \thcstressncrho
  }
  {
    \ppd{\gibbs}{\thcstressncrho}
    \left(
      \temp, \thcstressncrho
    \right)
    \fid{\overline{\thcstressncrho}}
  }
  -
  \divergence \hfluxc
  .
\end{equation}
(We are exploiting the fact that $\thcstressncrho$ commutes with $\pd{\gibbs}{\thcstressncrho}$ and the cyclic property of the trace. The fact that $\thcstressncrho$ commutes with $\pd{\gibbs}{\thcstressncrho}$ follows again from the representation theorems for isotropic functions.) Finally, the standard manipulation of the heat flux term yields
\begin{multline}
  \label{eq:45}
  \rho \dd{\entropy}{t}
  +
  \divergence
  \left(
    \frac{\hfluxc}{\temp}
  \right)
  \\
  =
  \frac{1}{\temp}
  \left\{
    \tensordot{
      \left\{
        \traceless{\cstress}
        +
        2
        \rho
        \traceless{
          \left[
            \thcstressncrho
            \ppd{\gibbs}{\thcstressncrho}
            \left(
              \temp, \thcstressncrho
            \right)
            \thcstressncrho
          \right]
        }
      \right\}
    }
    {\traceless{\gradsym}}
    +
    \rho
    \tensordot{
      \thcstressncrho
    }
    {
      \ppd{\gibbs}{\thcstressncrho}
      \left(
        \temp, \thcstressncrho
      \right)
      \fid{\overline{\thcstressncrho}}
    }
  \right\}
  -
  \frac{\vectordot{\hfluxc}{\nabla \temp}}{\temp^2}
  ,
\end{multline}
which is the formula used by~\cite{rajagopal.kr.srinivasa.ar:gibbs-potential-based}.

\subsubsection{Manipulations with entropy production -- indirect use of Gibbs free energy}
\label{sec:manip-with-entr-1}

Note however that~\eqref{eq:34} \emph{can be also manipulated differently}, and in this case \emph{no further structural assumptions on the specific Gibbs free energy are necessary}. In the incompressible setting~\eqref{eq:34} reads
\begin{equation}
  \label{eq:46}
  \rho \temp \dd{\entropy}{t}
  =
  \tensordot{\dcstresssymb}{\traceless{\gradsym}}
  -
  \tensordot{\thcstressnc}{\dd{\henckync}{t}}
  -
  \divergence \hfluxc
  .
\end{equation}
Unlike in the previous manipulation we keep the Hencky strain tensor~$\henckync$ in the equation, but we replace its material time derivative by the upper convected derivative,
\begin{equation}
  \label{eq:47}
  \fid{\overline{\henckync}}
  =
  \dd{\henckync}{t}
  -
  \gradvl
  \henckync
  -
  \henckync
  \transpose{\gradvl}
  .
\end{equation}
(The rationale for this manipulation is the same as in the case of the stress $\thcstressncrho$.) Next, using the standard manipulation of the heat flux $\hfluxc$, we see that~\eqref{eq:34} can be finally rewritten as
\begin{equation}
  \label{eq:48}
  \rho \dd{\entropy}{t}
  +
  \divergence
  \left(
    \frac{\hfluxc}{\temp}
  \right)
  =
  \frac{1}{\temp}
  \left\{
    \tensordot{
      \left\{
        \traceless{
          \cstress
        }
        -
        2
        \rho
        \traceless{
          \left[
            \thcstressncrho
            \henckync
          \right]
        }
      \right\}
    }
    {\traceless{\gradsym}}
    +
    \rho
    \tensordot{
      \thcstressncrho
    }
    {
      \fid{\overline{\henckync}}
    }
  \right\}
  -
  \frac{\vectordot{\hfluxc}{\nabla \temp}}{\temp^2}
  .
\end{equation}

\subsubsection{Constitutive relations}
\label{sec:const-relat}

Either~\eqref{eq:48} or~\eqref{eq:45} can \emph{serve as a starting point for the specification of constitutive relations}. Once we have postulated a formula for the Gibbs free energy, which is tantamount to the specification of the \emph{energy storage mechanism} in the material of interest, we can proceed with the specification of the \emph{entropy production mechanisms}. This boils down to the specification of the formula for the right-hand side of the generic entropy evolution equation that has the structure
\begin{equation}
  \label{eq:144}
  \rho \dd{\entropy}{t}
  +
  \divergence
  \entfluxc
  =
  \entprodc
  ,
\end{equation}
where $\entfluxc$ denotes the entropy flux and $\entprodc$ denotes the entropy production. Once the required entropy production---the right-hand side or entropy evolution equation---is specified, we compare it with the entropy production implied by the choice of the Gibbs free energy, that is with the genuine right hand-side of~\eqref{eq:48} or~\eqref{eq:45}, which will allow us to fix the constitutive equations for $\cstress$ and $\thcstress$. Examples of this procedure are given below in Section~\ref{sec:rederivation}. 

A more sophisticated version of this approach is based on the assumption of the maximisation of the entropy production, see~\cite{rajagopal.kr.srinivasa.ar:on*7}. This approach is worked out in Section~\ref{sec:entr-prod-maxim}.

\subsubsection{Temperature evolution equation}
\label{sec:temp-evol-equat}
Finally, let us remark that once the specific Gibbs free energy is given, one can then exploit the relation between the entropy and the specific Gibbs free energy, see~\eqref{eq:30}, and convert~\eqref{eq:45} or~\eqref{eq:48} respectively to an evolution equation for the temperature, indeed
\begin{equation}
  \label{eq:49}
  \dd{\entropy}{t}
  =
  -
  \dd{}{t}
  \left(
    \pd{\gibbs}{\temp} \left( \temp, \thcstressncrho \right)
  \right)
  =
  -
  \ppd{\gibbs}{\temp} \left( \temp, \thcstressncrho \right)
  \dd{\temp}{t}
  -
  \tensordot{
    \pd{^2\gibbs}{\temp \partial \thcstressncrho} \left( \temp, \thcstressncrho \right)
  }
  {
    \dd{\thcstressnc}{t}
  }
  .
\end{equation}
Note that if we use the assumption~\eqref{eq:41} that the mixed derivative of the specific Gibbs free energy vanishes, then the formula~\eqref{eq:49} simplifies to
$
\dd{\entropy}{t}
=
-
\ppd{\gibbs}{\temp} \left( \temp, \thcstressncrho \right)
\dd{\temp}{t}
$. If it is not the case, we need to consider the term $\tensordot{\pd{^2\gibbs}{\temp \partial \thcstressncrho} \left( \temp, \thcstressncrho \right)}{\dd{\thcstressnc}{t}}$  as well. The explicit formula for this term however depends on the rate-type constitutive relation for the stress tensor $\thcstressnc$.

\section{Example -- Giesekus/Oldroyd-B viscoelastic rate-type fluid in the approach based on the Gibbs free energy}
\label{sec:gies-visc-rate}

We now use the approach based on assuming the Gibbs free energy and appropriate entropy producing mechanisms to develop the popular incompressible rate-type viscoelastic fluid model due to Giesekus, see~\cite{giesekus.h:simple}. (Note that this model reduces, for a special choice of parameters, to the Oldroyd-B model, see~\cite{oldroyd.jg:on}.) In this case we have the luxury of writing down simple explicit formulas both for the Helmholtz free energy and the Gibbs free energy as well as for the entropy production. This makes the model ideal for clarifying the key ideas.

\subsection{Giesekus/Oldroyd-B model via the Helmholtz free energy -- classical approach}
\label{sec:giesekus-model-its}
Let us first summarise basic facts regarding the classical incompressible viscoelastic rate-type model derived by Giesekus for heat conducting fluids with constant specific heat at constant volume~$\cheatvolref$. (Note that the original derivation, see~\cite{giesekus.h:simple}, is done in the purely mechanical context; no temperature evolution equation is given. The corresponding temperature evolution equation is derived for example in~\cite{wapperom.p.hulsen.ma:thermodynamics}.) The governing equations for the mechanical quantities are in this case
\begin{subequations}
  \label{eq:giesekus-governing-equations}
  \begin{align}
    \label{eq:giesekus-incompressibility-condition}
    \divergence \vec{v} 
    &=
      0,
    \\
    \label{eq:giesekus-linear-momentum-balance}
    \rho \dd{\vec{v}}{t}
    &=
      \divergence \cstress + \rho \vec{b},
    \\
    \label{eq:giesekus-lcg-evolution-equation}
    \nu_1 \fid{\overline{\lcgnc}}
    &=
      -
      \mu(\theta)
      \left[
      \alpha \lcgnc^2 + (1 - 2 \alpha) \lcgnc - (1 - \alpha) \identity
      \right]
      ,
  \end{align}
  and the temperature evolution equation takes the form
  \begin{multline}
    \label{eq:giesekus-temperature-evolution-equation}
    \rho \cheatvolref
    \dd{\temp}{t}
    =
    \divergence \left( \kappa \nabla \temp \right)
    +
    2 \nu \tensordot{\traceless{\gradsym}}{\traceless{\gradsym}}
    +
    \frac{\mu(\theta)^2}{2 \nu_1} \Tr 
    \left[ 
      \alpha \lcgnc^2 + (1 - 3 \alpha) \lcgnc + (1 - \alpha) \inverse{\lcgnc} + (3 \alpha - 2) \identity
    \right]
    \\
    +
    \temp
    \dd{\mu}{\temp}
    \tensordot{\traceless{\gradsym}}{\traceless{\left(\lcgnc\right)}}
    -
    \temp
    \frac{\mu(\theta)}{2 \nu_1}
    \dd{\mu}{\temp}
    \Tr
    \left[
      \lcgnc
      +
      \inverse{\lcgnc}
      -
      2 \identity
    \right]
    .
  \end{multline}
  The Cauchy stress tensor $\cstress$ is given by the formulae
  \begin{equation}
    \label{eq:56}
    \cstress = \mns \identity + \traceless{\cstress},
    \qquad
    \traceless{\cstress} = 2 \nu \traceless{\gradsym} + \mu(\theta) \traceless{(\lcgnc)},
  \end{equation}
\end{subequations}
where $\alpha \in [0,1]$ is a model parameter, $\nu$, $\nu_1$ are positive material constants or functions that can eventually depend on temperature, and $\mu$ is a material constant or a function \emph{proportional to the thermodynamic temperature}; $\mu$ plays the role of the shear modulus for the elastic response, while the combination $\frac{\nu_1}{\mu(\temp)}$ plays the role of relaxation time, see~\eqref{eq:giesekus-lcg-evolution-equation}. (More complex temperature dependence of $\mu$ is not allowed. If we were dealing a more complex function than the constant/linear function of temperature, then we would get a model with specific heat at constant volume that depends on $\lcgnc$, see for example comments in~\cite{hron.j.milos.v.ea:on}. While such models might be useful, we do not consider them in the present contribution.) As we shall see later the material constant/function $\mu$ enters, unlike $\nu$ and $\nu_1$, the formula for the specific Helmholtz free energy $\fenergy$, whose derivatives with respect to~$\temp$ and~$\lcgnc$ determine the properties of the material. Therefore we shall explicitly write $\mu(\temp)$ in order to emphasise the temperature dependence of $\mu$. On the other hand, for $\nu$ and $\nu_1$ the temperature dependence is not---from the theoretical point of view---as important as in the case of $\mu$, hence we shall simply write $\nu$ and $\nu_1$ instead of $\nu(\temp)$ and $\nu_1(\temp)$. Finally, the symbol $\vec{b}$ denotes the body force.

A few remarks concerning the system~\eqref{eq:giesekus-governing-equations} are in order. First, we note that the right-hand side of~\eqref{eq:giesekus-lcg-evolution-equation} can be rewritten as
\begin{equation}
  \label{eq:57}
  -
  \mu(\theta)
  \left[
    \alpha \lcgnc^2 + (1 - 2 \alpha) \lcgnc - (1 - \alpha) \identity
  \right]
  =
  -
  \mu(\theta)
  \left[
    \left(
      \lcgnc - \identity
    \right)
    +
    \alpha
    \left(
      \lcgnc - \identity
    \right)^2
  \right]
  .
\end{equation}
In this form it is straightforward to see that if $\lcgnc = \identity$, then the right-hand side of~\eqref{eq:giesekus-lcg-evolution-equation} vanishes. (This is a useful information if one is interested in steady states predicted by~\eqref{eq:giesekus-governing-equations} in the case of zero external forcing.) Second, the last term in the temperature evolution equation~\eqref{eq:giesekus-temperature-evolution-equation} can be rewritten as
\begin{multline}
  \label{eq:58}
  \frac{\mu(\theta)^2}{2 \nu_1}
  \Tr 
  \left[ 
    \alpha \lcgnc^2 + (1 - 3 \alpha) \lcgnc + (1 - \alpha) \inverse{\lcgnc} + (3 \alpha - 2) \identity
  \right]
  \\
  =
  \frac{\mu(\theta)^2}{2 \nu_1}
  \Tr 
  \left[ 
    \inverse{\lcgnc}
    \left(
      \left(
        1 - \alpha
      \right)
      \identity
      +
      \alpha
      \lcgnc
    \right)
    \left(
      \lcgnc - \identity
    \right)^2
  \right]
  .
\end{multline}
This manipulation shows that the corresponding term is for $\alpha \in [0,1]$ nonnegative, and that it vanishes for $\lcgnc = \identity$, which a prospective steady state in the system without external forcing. Third, we recall that if we set $\alpha=0$, then we obtain the standard Oldroyd-B model, see~\cite{oldroyd.jg:on}.

The specific Helmholtz free energy $\fenergy$ is in this case of model~\eqref{eq:giesekus-governing-equations} known to be given by the formula
\begin{subequations}
  \label{eq:59}
  \begin{align}
    \label{eq:60}
    \fenergy(\temp, \lcgnc) &=_{\bydefinition} \fenergy_{\mathrm{thermal}}(\temp) + \frac{\mu(\temp)}{2\rho} \left( \Tr \lcgnc - 3 - \ln \det \lcgnc \right), \\
    \label{eq:61}
    \fenergy_{\mathrm{thermal}}(\temp) &=_{\bydefinition} - \cheatvolref \temp \left( \ln\frac{\temp}{\temp_{\reference}} - 1 \right).
  \end{align}
\end{subequations}
and the entropy production reads
\begin{equation}
  \label{eq:giesekus-entropy-production}
  \entprodc
  =
  \frac{1}{\temp}
  \left\{
    2 \nu \tensordot{\traceless{\gradsym}}{\traceless{\gradsym}}
    +
    \frac{\mu(\theta)^2}{2 \nu_1} \Tr 
    \left[ 
      \alpha \lcgnc^2 + (1 - 3 \alpha) \lcgnc + (1 - \alpha) \inverse{\lcgnc} + (3 \alpha - 2) \identity
    \right]
  \right\}
  +
  \kappa
  \frac{
    \vectordot{\nabla \temp}{\nabla \temp}
  }
  {\temp^2}
  .
\end{equation}

Note that~\eqref{eq:58} shows that the entropy production is nonnegative. We recall that the Hencky strain tensor associated to the elastic response $\henckync$ is given by~\eqref{eq:10}, hence we have
\begin{equation}
  \label{eq:62}
  \lcgnc = \exponential{2 \henckync},
\end{equation}
and the formula for the Helmholtz free energy $\fenergy$ in terms of $\henckync$ reads
\begin{equation}
  \label{eq:63}
  \fenergy(\temp, \henckync) =_{\bydefinition} \fenergy_{\mathrm{thermal}}(\temp) + \frac{\mu(\temp)}{2\rho} \left( \Tr \exponential{2 \henckync} - 3 - \ln \det \exponential{2 \henckync} \right).   
\end{equation}
Using the fact that the Hencky strain tensor is represented by a symmetric matrix and the standard identity $\det \exponential{\generictensor} = \exponential{\Tr \generictensor}$, we see that~\eqref{eq:63} can be rewritten as
\begin{equation}
  \label{eq:64}
  \fenergy(\temp, \henckync) =_{\bydefinition} \fenergy_{\mathrm{thermal}}(\temp) + \frac{\mu(\temp)}{2\rho} \left( \Tr \exponential{2 \henckync} - 3 - 2 \Tr \henckync \right).   
\end{equation}
The stress $\thcstressnc$ is then given via the definition $\frac{\thcstressnc}{\rho} = \pd{\fenergy}{\henckync}(\temp, \henckync)$, see~\eqref{eq:22}, which in our case yields
\begin{equation}
  \label{eq:65}
  \thcstressnc
  =
  \mu (\temp)
  \left(
    \exponential{2 \henckync}
    -
    \identity
  \right)
  ,
\end{equation}
where we have used Lemma~\ref{lemma:differentiation}. Equation~\eqref{eq:65} can be solved explicitly for the Hencky strain $\henckync$,
\begin{equation}
  \label{eq:66}
  \henckync
  =
  \frac{1}{2}
  \ln
  \left(
    \frac{\thcstressnc}{\mu(\temp)}
    +
    \identity
  \right)
  .
\end{equation}
Note that if we rewrite \eqref{eq:65} in terms of~$\lcgnc$, then we get
\begin{equation}
  \label{eq:67}
  \thcstressnc = \mu(\temp) \left(\lcgnc - \identity \right). 
\end{equation}
Having obtained formula~\eqref{eq:66} we can use the definition of the specific Gibbs free energy $\gibbs = \fenergy - \tensordot{\frac{\thcstressnc}{\rho}}{\henckync}$, and we can  write down an explicit formula for the Gibbs free energy corresponding to the Helmholtz free energy introduced in~\eqref{eq:59},
\begin{equation}
  \label{eq:68}
  \gibbs(\temp, \thcstressnc)
  =
  \gibbs_{\mathrm{thermal}}(\temp)
  +
  \frac{\mu(\temp)}{2 \rho}
  \left(
    \Tr
    \left(
      \frac{\thcstressnc}{\mu(\temp)}
    \right)
    -
    \Tr
    \ln
    \left(
      \frac{\thcstressnc}{\mu(\temp)}
      +
      \identity
    \right)
  \right)
  -
  \frac{1}{2 \rho}
  \tensordot{
    \thcstressnc
  }
  {
    \ln
    \left(
      \frac{\thcstressnc}{\mu(\temp)}
      +
      \identity
    \right)
  }
  ,
\end{equation}
which is easy to convert to the final form
\begin{equation}
  \label{eq:69}
  \gibbs(\temp, \thcstressncrho)
  =
  \gibbs_{\mathrm{thermal}}(\temp)
  +
  \frac{\mu(\temp)}{2 \rho}
  \Tr
  \left(
    \frac{\rho}{\mu(\temp)} \thcstressncrho
    -
    \left[
      \frac{\rho}{\mu(\temp)} \thcstressncrho
      +
      \identity
    \right]
    \ln
    \left[
      \frac{\rho}{\mu(\temp)} \thcstressncrho
      +
      \identity
    \right]
  \right)
  .
\end{equation}
Since the Legendre transformation from the Helmholtz free energy to the Gibbs free energy is done with respect to the deformation and the temperature variable is kept intact, we see that the function $\gibbs_{\mathrm{thermal}}(\temp)$ is identical to the function $\fenergy_{\mathrm{thermal}}(\temp)$.

We note that if we use spectral representation of~$\thcstressncrho$, then~\eqref{eq:68} can be rewritten as
\begin{equation}
  \label{eq:70}
  \gibbs(\temp, \thcstressncrho)
  =
  \gibbs_{\mathrm{thermal}}(\temp)
  +
  \frac{\mu(\temp)}{2 \rho}
  \sum_{i=1}^3
  \left(
    \frac{\rho}{\mu(\temp)} \lambda_i
    -
    \left[
      \frac{\rho}{\mu(\temp)} \lambda_i
      +
      1
    \right]
    \ln
    \left[
      \frac{\rho}{\mu(\temp)} \lambda_i
      +
      1
    \right]
  \right)
  ,
\end{equation}
where $\left\{\lambda_i\right\}_{i=1}^3$ are eigenvalues of $\thcstressncrho$. The function $f(x) = x - (x+1) \ln (x+1)$ that appears on the right-hand side of~\eqref{eq:70} is a nonpositive \emph{concave} function that vanishes if and only if $x=0$. Furthermore, by appealing to~\eqref{eq:69} we see that the Gibbs free energy becomes infinite once the quantity 
$
\frac{\thcstressnc}{\mu(\temp)}
+
\identity
$
vanishes. This means that the quantity
$
\frac{\thcstressnc}{\mu(\temp)}
+
\identity
$
remains a positive definite matrix. (This is equivalent to the positive definiteness of $\lcgnc$ in the classical approach.) Finally, we observe that the Gibbs free energy vanishes for $\thcstressnc = \tensorzero$, which corresponds to the fact that the Helholtz free energy vanishes for~$\lcgnc = \identity$.

\subsection{Approach based on Gibbs free energy}
\label{sec:rederivation}
Let us now assume that the Gibbs free energy $\gibbs$ is given by~\eqref{eq:69}, and let us \emph{rederive} the governing equations~\eqref{eq:giesekus-governing-equations} within the approach based on the Gibbs free energy. (Instead of the Helmholtz free energy we start with the Gibbs free energy, and we observe what steps need to be taken, if we want to rederive the Giesekus model.) Using the standard matrix calculus we can observe that
\begin{equation}
  \label{eq:71}
  \pd{\gibbs}{\thcstressncrho}(\temp, \thcstressncrho)
  =
  -
  \frac{1}{2}
  \ln
  \left(
    \frac{\rho}{\mu(\temp)}\thcstressncrho
    +
    \identity
  \right)
  ,
\end{equation}
which is the expected result, since the use of Legendre transformation implies that $\henckync = - \pd{\gibbs}{\thcstressncrho}(\temp, \thcstressncrho)$. Compare~\eqref{eq:71} and~\eqref{eq:66}.

Now we need to exploit the evolution equation for the entropy. In particular, we need to propose an evolution equation for $\thcstressnc$ and a constitutive relation for $\cstress$ in such a way that the entropy production is nonnegative. (The thermal part---the heat flux---can be manipulated in the standard manner.)
\begin{equation}
  \label{eq:72}
  \rho \dd{\entropy}{t}
  +
  \divergence
  \left(
    \frac{\hfluxc}{\temp}
  \right)
  =
  \frac{1}{\temp}
  \left\{
    \tensordot{\traceless{\cstress}}{\traceless{\gradsym}}
    +
    \tensordot{
      \thcstressnc
    }
    {
      \dd{\henckync}{t}
    }
  \right\}
  -
  \frac{\vectordot{\hfluxc}{\nabla \temp}}{\temp^2}
  .
\end{equation}
The term
$
\tensordot{
  \thcstressnc
}
{
  \dd{\henckync}{t}
}
$
can be expressed as
\begin{multline}
  \label{eq:73}
  \tensordot{
    \thcstressnc
  }
  {
    \dd{\henckync}{t}
  }
  =
  \tensordot{
    \thcstressnc
  }
  {
    \exponential{-2 \henckync}
    \exponential{2 \henckync}
    \dd{\henckync}{t}
  }
  =
  \frac{1}{2}
  \tensordot{
    \thcstressnc
  }
  {
    \exponential{-2 \henckync}
    \dd{}{t}
    \exponential{2 \henckync}
  }
  \\
  =
  \frac{1}{2}
  \tensordot{
    \thcstressnc
    \exponential{-2 \henckync}
  }
  {
    \left(
      \fid{
        \overline{
          \exponential{2 \henckync}
        }
      }
      +
      \gradvl
      \exponential{2 \henckync}
      +
      \exponential{2 \henckync}
      \transpose{\gradvl}
    \right)
  }
  =
  \tensordot{\thcstressnc}{\gradsym}
  +
  \frac{1}{2}
  \tensordot{
    \thcstressnc
    \exponential{-2 \henckync}
  }
  {
    \fid{
      \overline{
        \exponential{2 \henckync}
      }
    }
  }
  ,
\end{multline}
while the motivation for this manipulation is to get the term $\tensordot{\thcstressnc}{\gradsym}$, which can be conveniently manipulated in~\eqref{eq:72}. In~\eqref{eq:73} we are using the cyclic property of the trace and the fact that $\thcstressnc$ and $\henckync$ commute, which implies that $\thcstressnc$ and $\exponential{-2\henckync}$ commute as well. If we substitute~\eqref{eq:73} into~\eqref{eq:72}, we get
\begin{equation}
  \label{eq:74}
  \rho \dd{\entropy}{t}
  +
  \divergence
  \left(
    \frac{\hfluxc}{\temp}
  \right)
  =
  \frac{1}{\temp}
  \left\{
    \tensordot{\left[ \traceless{\cstress} - \devthcstressncdelta \right]}{\traceless{\gradsym}}
    -
    \frac{1}{2}
    \tensordot{
      \thcstressnc
      \exponential{-2 \henckync}
    }
    {
      \fid{
        \overline{
          \exponential{2 \henckync}
        }
      }
    }
  \right\}
  -
  \frac{\vectordot{\hfluxc}{\nabla \temp}}{\temp^2}
  .
\end{equation}

The primitive quantity in the approach based on the Gibbs free energy is the reduced stress tensor $\thcstressncrho = \frac{\thcstressnc}{\rho}$, and we need to find an evolution equation for this quantity. (Or for that matters an evolution equation for $\thcstressnc$ since these quantities differ only by the multiplication by the density $\rho$, which is in our case a constant.) We recall that the evolution equation for $\thcstressnc$ and the constitutive relation for $\cstress$ must be chosen in such a way that the corresponding terms on the right-hand side of~\eqref{eq:72} are nonnegative.

\subsubsection{Oldroyd-B model}
\label{sec:oldroyd-b-model}

Before we discuss the Giesekus model, we focus on the case $\alpha=0$ which leads to the Oldroyd-B model. Regarding the entropy production terms in~\eqref{eq:74} the following simple choice 
\begin{subequations}
  \label{eq:75}
  \begin{align}
    \label{eq:76}
    \thcstressnc &=_{\bydefinition}
                   -
                   \nu_1
                   \fid{
                   \overline{
                   \exponential{2 \henckync}
                   }
                   },
    \\
    \label{eq:77}
    \traceless{\cstress} &=_{\bydefinition} \devthcstressncdelta + 2 \nu \traceless{\gradsym},
  \end{align}
\end{subequations}
is of particular interest. If we use~\eqref{eq:75} in~\eqref{eq:74}, then we get
\begin{equation}
  \label{eq:78}
  \rho \dd{\entropy}{t}
  +
  \divergence
  \left(
    \frac{\hfluxc}{\temp}
  \right)
  =
  \frac{1}{\temp}
  \left\{
    2 \nu
    \tensordot{\traceless{\gradsym}}{\traceless{\gradsym}}
    +
    \frac{\nu_1}{2}
    \tensordot{
      \fid{
        \overline{
          \exponential{2 \henckync}
        }
      }
      \exponential{-2 \henckync}
    }
    {
      \fid{
        \overline{
          \exponential{2 \henckync}
        }
      }
    }
  \right\}
  -
  \frac{\vectordot{\hfluxc}{\nabla \temp}}{\temp^2}
  ,
\end{equation}
hence the entropy production in mechanical processes in nonnegative as required. Indeed, the term in the entropy production can be rewritten as
\begin{multline}
  \label{eq:79}
  \tensordot{
    \fid{
      \overline{
        \exponential{2 \henckync}
      }
    }
    \exponential{-2 \henckync}
  }
  {
    \fid{
      \overline{
        \exponential{2 \henckync}
      }
    }
  }
  =
  \Tr
  \left[
    \fid{
      \overline{
        \exponential{2 \henckync}
      }
    }
    \exponential{-2 \henckync}
    \fid{
      \overline{
        \exponential{2 \henckync}
      }
    }
  \right]
  \\
  =
  \Tr
  \left[
    \fid{
      \overline{
        \exponential{2 \henckync}
      }
    }
    \exponential{-\henckync}
    \exponential{-\henckync}
    \fid{
      \overline{
        \exponential{2 \henckync}
      }
    }
  \right]
  =
  \tensordot{
    \fid{
      \overline{
        \exponential{2 \henckync}
      }
    }
    \exponential{-\henckync}
  }
  {
    \fid{
      \overline{
        \exponential{2 \henckync}
      }
    }
    \exponential{-\henckync}
  }
  ,
\end{multline}
or as
\begin{equation}
  \label{eq:80}
  \frac{\nu_1}{2}
  \tensordot{
    \fid{
      \overline{
        \exponential{2 \henckync}
      }
    }
    \exponential{-2 \henckync}
  }
  {
    \fid{
      \overline{
        \exponential{2 \henckync}
      }
    }
  }
  =
  \frac{1}{2 \nu_1}
  \tensordot{
    \thcstressnc
    \exponential{-2 \henckync}
  }
  {
    \thcstressnc
  }
  ,
\end{equation}
which makes its nonnegativity self-evident. (We have exploited the symmetry of all involved matrices.) The term can be also understood as a weighted scalar product of matrices
$
\fid{
  \overline{
    \exponential{2 \henckync}
  }
}
$
with the weight $\exponential{-2\henckync} = \inverse{\left( \frac{\thcstressnc}{\mu(\temp)} + \identity \right)}$, where we have used the fact that $\henckync = - \pd{\gibbs}{\thcstressncrho}(\temp, \thcstressncrho)$ and the particular formula for the Gibbs free energy, see~\eqref{eq:71}.

Note that the last observation implies that if we use constitutive assumptions~\eqref{eq:75} and~\eqref{eq:71}, then the entropy evolution equation~\eqref{eq:78} rewritten in terms of~$\thcstressnc$ reads
\begin{equation}
  \label{eq:81}
  \rho \dd{\entropy}{t}
  +
  \divergence
  \left(
    \frac{\hfluxc}{\temp}
  \right)
  =
  \frac{1}{\temp}
  \left\{
    2 \nu
    \tensordot{\traceless{\gradsym}}{\traceless{\gradsym}}
    +
    \frac{1}{2 \nu_1}
    \tensordot{
      \thcstressnc
      \inverse{
        \left(
          \frac{\thcstressnc}{\mu(\temp)}
          +
          \identity
        \right)
      }
    }
    {
      \thcstressnc
    }
  \right\}
  -
  \frac{\vectordot{\hfluxc}{\nabla \temp}}{\temp^2}
  .
\end{equation}
In this case the nonnegativity of the entropy production in mechanical processes still holds---equation~\eqref{eq:81} is the same as equation~\eqref{eq:78} but rewritten in terms of different variables. However, the nonnegativity of the entropy production is less straightforward to recognise than in the formulation based on $\henckync$.

If we further assume that the heat flux is given by the standard Fourier law,
\begin{equation}
  \label{eq:82}
  \hfluxc = - \kappa \nabla \temp,
\end{equation}
where $\kappa$ is a nonnegative constant, then the entropy production both in thermal and mechanical processes is nonnegative and the second law of thermodynamics holds.

Once we have identified the energy storage mechanisms---specific Gibbs free energy~\eqref{eq:69}---and the entropy production mechanisms---terms on the right-hand side of the entropy evolution equation~\eqref{eq:78}---the complete system of governing equations is obtained as follows. First, we use~\eqref{eq:76} and~\eqref{eq:65}, and we can conclude that the evolution equation for $\henckync$ reads
\begin{equation}
  \label{eq:83}
  -
  \nu_1
  \fid{
    \overline{
      \exponential{2 \henckync}
    }
  }
  =
  \mu (\temp)
  \left(
    \exponential{2 \henckync}
    -
    \identity
  \right)
  ,
\end{equation}
which can be as well rewritten in terms of $\thcstressnc$ as
\begin{equation}
  \label{eq:84}
  \nu_1
  \fid{
    \overline{
      \left(
        \frac{
          \thcstressnc
        }
        {
          \mu(\temp)
        }
      \right)
    }
  }
  +
  \thcstressnc
  =
  2 \nu_1 \gradsym.
\end{equation}
(In the latter case it suffices to take the upper convected derivative of~\eqref{eq:65} and use~\eqref{eq:76}. We recall that the definition of the upper convected derivative implies that $\fid{\identity} = - 2 \gradsym$.) Finally, if we wanted to, we could go back to the left Cauchy--Green tensor $\lcgnc$, see~\eqref{eq:62}, which allows us to rewrite~\eqref{eq:83} as
\begin{equation}
  \label{eq:85}
  \nu_1 \fid{\overline{\lcgnc}}
  =
  -
  \mu(\theta)
  \left(
    \lcgnc - \identity
  \right)
  .
\end{equation}
This is the evolution equation~\eqref{eq:giesekus-lcg-evolution-equation} for $\alpha = 0$.

If we use~\eqref{eq:84}, then the system of governing equations for the mechanical quantities reads
\begin{subequations}
  \label{eq:oldroyd-b-gibbs-governing-equations}
  \begin{align}
    \label{eq:oldroyd-b-gibbs-incompressibility-condition}
    \divergence \vec{v} 
    &=
      0,
    \\
    \label{eq:oldroyd-b-gibbs-linear-momentum-balance}
    \rho \dd{\vec{v}}{t}
    &=
      \divergence \cstress + \rho \vec{b},
    \\
    \label{eq:oldroyd-b-gibbs-lcg-evolution-equation}
    \nu_1
    \fid{
    \overline{
    \left(
    \frac{
    \thcstressnc
    }
    {
    \mu(\temp)
    }
    \right)
    }
    }
    +
    \thcstressnc
    &=
      2 \nu_1 \gradsym
      ,
  \end{align}
  where the Cauchy stress tensor $\cstress$ is given by the formula
  \begin{equation}
    \label{eq:86}
    \cstress = \mns \identity +  2 \nu \traceless{\gradsym} + \devthcstressncdelta.
  \end{equation}
\end{subequations}
Since~\eqref{eq:oldroyd-b-gibbs-lcg-evolution-equation} can be rewritten as~\eqref{eq:85}, and since \eqref{eq:65} holds, that is
$
\thcstressnc
=
\mu (\temp)
\left(
  \exponential{2 \henckync}
  -
  \identity
\right)$
holds, we see that~\eqref{eq:oldroyd-b-gibbs-governing-equations} is equivalent to the mechanical part of~\eqref{eq:giesekus-governing-equations}. (A minor redefinition of the mean normal stress $\mns$ is necessary. This is however an inconsequential modification of the governing equations.)

Regarding the temperature evolution equation~\eqref{eq:giesekus-temperature-evolution-equation} we have to start with the entropy evolution equation either in the form~\eqref{eq:78} or in the form~\eqref{eq:81}. In particular, we need to show that the entropy production identified in~\eqref{eq:78} or equivalently in~\eqref{eq:81} is tantamount to the entropy production~\eqref{eq:giesekus-entropy-production} that leads to the Oldroyd-B model. (We recall that Oldroyd-B model corresponds to $\alpha=0$.) The critical term can be manipulated as follows
\begin{multline}
  \label{eq:87}
  \frac{\nu_1}{2}
  \tensordot{
    \fid{
      \overline{
        \exponential{2 \henckync}
      }
    }
    \exponential{-2 \henckync}
  }
  {
    \fid{
      \overline{
        \exponential{2 \henckync}
      }
    }
  }
  =
  \frac{\nu_1}{2}
  \tensordot{
    \fid{
      \overline{
        \lcgnc
      }
    }
    \inverse{\lcgnc}
  }
  {
    \fid{
      \overline{
        \lcgnc
      }
    }
  }
  =
  \frac{\nu_1}{2}
  \tensordot{
    \frac{\mu(\theta)}{\nu_1}
    \left(
      \lcgnc - \identity
    \right)
    \inverse{\lcgnc}
  }
  {
    \frac{\mu(\theta)}{\nu_1}
    \left(
      \lcgnc - \identity
    \right)
  }
  \\
  =
  \frac{\mu(\theta)^2}{2 \nu_1}
  \Tr
  \left[
    \inverse{\lcgnc}
    \left(
      \lcgnc - \identity
    \right)^2
  \right]
  .
\end{multline}
Making use of identity~\eqref{eq:58} we see that the entropy production mechanisms postulated in~\eqref{eq:78} or~\eqref{eq:81} are identical to those postulated in the standard derivation of the Oldroyd-B model~\eqref{eq:giesekus-entropy-production}. Note that~\eqref{eq:87} explicitly shows that the entropy production for the Oldroyd-B fluid can be equivalently written down either in terms of a algebraic formula for $\lcgnc$ or in terms of algebraic formula to the rates $\fid{\overline{\lcgnc}}$.

Since the entropy production mechanisms $\entprodc$ are now confirmed to be the same in both approaches, we see that the right-hand side of the entropy evolution equation
\begin{equation}
  \label{eq:88}
  \rho \dd{\entropy}{t}
  +
  \divergence
  \left(
    \frac{\hfluxc}{\temp}
  \right)
  =
  \entprodc
\end{equation}
is the same both for the model derived via the Gibbs free energy and the Helmholtz free energy. Finally, it suffices to recall that the the entropy $\entropy$ is given in terms of derivatives of the potentials as $\entropy(\temp, \thcstressncrho) = - \pd{\gibbs}{\temp}(\temp, \thcstressncrho)$ and $\entropy(\temp, \lcgnc) = - \pd{\fenergy}{\temp}(\temp, \lcgnc)$, see~\eqref{eq:38} and~\eqref{eq:20}. Consequently, if we substitute for $\entropy$ into~\eqref{eq:88}, we get for both approaches the same evolution equation for the temperature. (In the case of the Gibbs free energy the primitive variable will be the reduced stress tensor $\thcstressncrho$, while in the case of Helmholtz free energy the primitive variable will be the left Cauchy--Green tensor associated to the elastic response $\lcgnc$.) For readers convenience both approaches are compared in Summary~\ref{summary:gibbs} and Summary~\ref{summary:helmholtz}.

\begin{summary}[Incompressible Oldroyd-B model via Gibbs free energy]
  \label{summary:gibbs}
  Specific Gibbs free energy:
  \begin{align*}
    \gibbs(\temp, \thcstressncrho)
    &=_{\bydefinition}
      \gibbs_{\mathrm{thermal}}(\temp)
      +
      \frac{\mu(\temp)}{2 \rho}
      \Tr
      \left(
      \frac{\rho}{\mu(\temp)} \thcstressncrho
      -
      \left[
      \frac{\rho}{\mu(\temp)} \thcstressncrho
      +
      \identity
      \right]
      \ln
      \left[
      \frac{\rho}{\mu(\temp)} \thcstressncrho
      +
      \identity
      \right]
      \right)
    \\
    \gibbs_{\mathrm{thermal}}(\temp)
    &=_{\bydefinition}
      - \cheatvolref \temp \left( \ln\frac{\temp}{\temp_{\reference}} - 1 \right)
  \end{align*}
  Reduced stress (notation):
  \begin{equation*}
    \thcstressncrho =_{\bydefinition} \frac{\thcstressnc}{\rho}
  \end{equation*}
  Entropy production:
  \begin{equation*}
    \entprodc
    =
    \frac{1}{\temp}
    \left\{
      2 \nu
      \tensordot{\traceless{\gradsym}}{\traceless{\gradsym}}
      +
      \frac{1}{2 \nu_1}
      \tensordot{
        \thcstressnc
        \inverse{
          \left(
            \frac{\thcstressnc}{\mu(\temp)}
            +
            \identity
          \right)
        }
      }
      {
        \thcstressnc
      }
    \right\}
    +
    \kappa
    \frac{\vectordot{\nabla \temp}{\nabla \temp}}{\temp^2}
  \end{equation*}
  Material parameters: $\cheatvolref$ specific heat at constant volume -- positive constant; $\mu(\temp)$ shear modulus -- nonnegative function, typically proportional to $\temp$; $\nu$, $\nu_1$ viscosity -- nonnegative functions of the primitive variables, typically constants; $\kappa$ thermal conductivity -- nonnegative function of the primitive variables, typically constant\bigskip

  Evolution equations (mechanical variables $\mns$, $\vec{v}$, $\thcstressnc$ and thermal variable $\temp$):
  \begin{align*}
    \divergence \vec{v} 
    &=
      0
    \\
    \rho \dd{\vec{v}}{t}
    &=
      \divergence \cstress + \rho \vec{b}
    \\
    \nu_1
    \fid{
    \overline{
    \left(
    \frac{
    \thcstressnc
    }
    {
    \mu(\temp)
    }
    \right)
    }
    }
    +
    \thcstressnc
    &=
      2 \nu_1 \gradsym
    \\
    -
    \rho
    \temp
    \dd{}{t}\left( \pd{\gibbs}{\temp}(\temp, \thcstressncrho) \right)
    &=
      \divergence
      \left(
      \kappa \temp
      \right)
      +
      2 \nu
      \tensordot{\traceless{\gradsym}}{\traceless{\gradsym}}
      +
      \frac{1}{2 \nu_1}
      \tensordot{
      \thcstressnc
      \inverse{
      \left(
      \frac{\thcstressnc}{\mu(\temp)}
      +
      \identity
      \right)
      }
      }
      {
      \thcstressnc
      }
  \end{align*}
  Cauchy stress tensor:
  \begin{equation*}
    \cstress = \mns \identity +  2 \nu \traceless{\gradsym} + \devthcstressncdelta
  \end{equation*}
  Thermodynamical relations:
  \begin{align*}
    \henckync &= - \pd{\gibbs}{\thcstressncrho}(\temp, \thcstressncrho) \\
    \entropy &= - \pd{\gibbs}{\temp}(\temp, \thcstressncrho)
  \end{align*}

  Left Cauchy--Green tensor $\lcgnc$ and Hencky strain tensor $\henckync$ associated to the elastic response:
  \begin{equation*}
    \lcgnc = \exponential{2 \henckync}
  \end{equation*}
\end{summary}

\begin{summary}[Incompressible Oldroyd-B model via Helmholtz free energy]
  \label{summary:helmholtz}
  Specific Helmholtz free energy:
  \begin{align*}
    \fenergy(\temp, \lcgnc) &=_{\bydefinition} \fenergy_{\mathrm{thermal}}(\temp) + \frac{\mu(\temp)}{2\rho} \left( \Tr \lcgnc - 3 - \ln \det \lcgnc \right)
    \\
    \fenergy_{\mathrm{thermal}}(\temp) &=_{\bydefinition} - \cheatvolref \temp \left( \ln\frac{\temp}{\temp_{\reference}} - 1 \right).
  \end{align*}
  Entropy production:
  \begin{equation*}
    \entprodc
    =
    \frac{1}{\temp}
    \left\{
      2 \nu \tensordot{\traceless{\gradsym}}{\traceless{\gradsym}}
      +
      \frac{\mu(\theta)^2}{2 \nu_1} \Tr 
      \left[ 
        \lcgnc + \inverse{\lcgnc} - 2 \identity
      \right]
    \right\}
    +
    \kappa
    \frac{\vectordot{\nabla \temp}{\nabla \temp}}{\temp^2}
  \end{equation*}
  Material parameters: $\cheatvolref$ specific heat at constant volume -- positive constant; $\mu(\temp)$ shear modulus -- nonnegative function, typically proportional to $\temp$; $\nu$, $\nu_1$ viscosity -- nonnegative functions of the primitive variables, typically constants; $\kappa$ thermal conductivity -- nonnegative function of the primitive variables, typically constant\bigskip

  Evolution equations (mechanical variables $\mns$, $\vec{v}$, $\lcgnc$ and thermal variable $\temp$):
  \begin{align*}
    \divergence \vec{v} 
    &=
      0
    \\
    \rho \dd{\vec{v}}{t}
    &=
      \divergence \cstress + \rho \vec{b}
    \\
    \nu_1 \fid{\overline{\lcgnc}}
    &=
      -
      \mu(\theta)
      \left[
      \lcgnc - \identity
      \right]
    \\
    -
    \rho
    \temp
    \dd{}{t}\left( \pd{\fenergy}{\temp}(\temp, \lcgnc) \right)
    &=
      \divergence
      \left(
      \kappa \temp
      \right)
      +
      2 \nu
      \tensordot{\traceless{\gradsym}}{\traceless{\gradsym}}
      +
      \frac{\mu(\theta)^2}{2 \nu_1} \Tr 
      \left[ 
      \lcgnc + \inverse{\lcgnc} - 2 \identity
      \right]
  \end{align*}
  Cauchy stress tensor:
  \begin{equation*}
    \cstress = \mns \identity +  2 \nu \traceless{\gradsym} + \mu(\theta) \traceless{\left( \lcgnc \right)}
  \end{equation*}
  Thermodynamical relations:
  \begin{align*}
    \entropy = - \pd{\fenergy}{\temp}(\temp, \lcgnc)
  \end{align*}

  Left Cauchy--Green tensor $\lcgnc$ and Hencky strain tensor $\henckync$ associated to the elastic response:
  \begin{equation*}
    \lcgnc = \exponential{2 \henckync}
  \end{equation*}

\end{summary}

\subsubsection{Giesekus model}
\label{sec:giesekus}

Regarding the Giesekus model, the derivation in the approach based on the Gibbs free energy is a slight modification of the derivation of the Oldroyd-B model presented above. The Gibbs free energy remains the same, see~\eqref{eq:69}, and the entropy evolution equation remains the same~\eqref{eq:74} as well, that is
\begin{equation}
  \label{eq:89}
  \rho \dd{\entropy}{t}
  +
  \divergence
  \left(
    \frac{\hfluxc}{\temp}
  \right)
  =
  \frac{1}{\temp}
  \left\{
    \tensordot{\left[ \traceless{\cstress} - \devthcstressncdelta \right]}{\traceless{\gradsym}}
    -
    \frac{1}{2}
    \tensordot{
      \thcstressnc
      \exponential{-2 \henckync}
    }
    {
      \fid{
        \overline{
          \exponential{2 \henckync}
        }
      }
    }
  \right\}
  -
  \frac{\vectordot{\hfluxc}{\nabla \temp}}{\temp^2}
  .
\end{equation}
The entropy production mechanisms are however modified to
\begin{equation}
  \label{eq:90}
  \rho \dd{\entropy}{t}
  +
  \divergence
  \left(
    \frac{\hfluxc}{\temp}
  \right)
  =
  \frac{1}{\temp}
  \left\{
    2 \nu
    \tensordot{\traceless{\gradsym}}{\traceless{\gradsym}}
    +
    \frac{1}{2 \nu_1}
    \tensordot{
      \thcstressnc
      \exponential{-2 \henckync}
    }
    {
      \left[
        \thcstressnc
        \left(
          \identity
          +
          \frac{\alpha}{\mu(\temp)}
          \thcstressnc
        \right)
      \right]
    }
  \right\}
  +
  \kappa
  \frac{\vectordot{\nabla \temp}{\nabla \temp}}{\temp^2}
  ,
\end{equation}
which implies constitutive relations
\begin{subequations}
  \label{eq:91}
  \begin{align}
    \label{eq:92}
    \hfluxc &=_{\bydefinition} - \kappa \nabla \temp, \\
    \label{eq:93}
    \thcstressnc
    \left(
    \identity
    +
    \frac{\alpha}{\mu(\temp)}
    \thcstressnc
    \right)
            &=_{\bydefinition}
              -
              \nu_1
              \fid{
              \overline{
              \exponential{2 \henckync}
              }
              },
    \\
    \label{eq:94}
    \traceless{\cstress} &=_{\bydefinition} \devthcstressncdelta + 2 \nu \traceless{\gradsym}.
  \end{align}
\end{subequations}

It is straightforward to check that the entropy production specified in~\eqref{eq:90} is nonnegative. Indeed, the critical term can be manipulated as follows
\begin{multline}
  \label{eq:95}
  \tensordot{
    \thcstressnc
    \exponential{-2 \henckync}
  }
  {
    \left[
      \thcstressnc
      \left(
        \identity
        +
        \frac{\alpha}{\mu(\temp)}
        \thcstressnc
      \right)
    \right]
  }
  =
  \tensordot{
    \thcstressnc
    \exponential{-2 \henckync}
  }
  {
    \thcstressnc
    \left[
      \alpha
      \left(
        \identity
        +
        \frac{\thcstressnc}{\mu(\temp)}
      \right)
      +
      \left(
        1 - \alpha
      \right)
      \identity
    \right]
  }
  \\
  =
  \left(
    1 - \alpha
  \right)
  \tensordot{
    \thcstressnc
    \exponential{-2 \henckync}
  }
  {
    \thcstressnc
  }
  +
  \alpha
  \tensordot{
    \thcstressnc
    \exponential{-2 \henckync}
  }
  {
    \left(
      \identity
      +
      \frac{\thcstressnc}{\mu(\temp)}
    \right)
    \thcstressnc
  }
  \\
  =
  \left(
    1 - \alpha
  \right)
  \tensordot{
    \thcstressnc
    \inverse{
      \left(
        \frac{\thcstressnc}{\mu(\temp)}
        +
        \identity
      \right)
    }
  }
  {
    \thcstressnc
  }
  +
  \alpha
  \tensordot{
    \thcstressnc
  }
  {
    \thcstressnc
  }
  ,
\end{multline}
where we have used the relation \eqref{eq:66} between the stress tensor~$\thcstressnc$ and the Hencky strain tensor~$\henckync$. The first term on the right-hand side of~\eqref{eq:95} is up to a positive factor $(1-\alpha)$ the same as the term in the entropy production for the Oldroyd-B model, see~\eqref{eq:81}, while the nonnegativity of the term
$
\alpha
\tensordot{
  \thcstressnc
}
{
  \thcstressnc
}
$
is evident. (Note that the entropy production is particularly simple if we set $\alpha=1$.) Using the same steps as in the previous section it is then straightforward to check that the constitutive relations~\eqref{eq:91} lead, once rewritten in terms of $\lcgnc$, to the governing equations for the standard Giesekus viscoelastic rate-type fluid.

To conclude, we see that the Giesekus viscoelastic rate-type fluid is specified by the Gibbs free energy in the form
\begin{align}
  \label{eq:96}
  \gibbs(\temp, \thcstressncrho)
  &=_{\bydefinition}
    \gibbs_{\mathrm{thermal}}(\temp)
    +
    \frac{\mu(\temp)}{2 \rho}
    \Tr
    \left(
    \frac{\rho}{\mu(\temp)} \thcstressncrho
    -
    \left[
    \frac{\rho}{\mu(\temp)} \thcstressncrho
    +
    \identity
    \right]
    \ln
    \left[
    \frac{\rho}{\mu(\temp)} \thcstressncrho
    +
    \identity
    \right]
    \right)
  \\
  \label{eq:97}
  \gibbs_{\mathrm{thermal}}(\temp)
  &=_{\bydefinition}
    - \cheatvolref \temp \left( \ln\frac{\temp}{\temp_{\reference}} - 1 \right)
\end{align}
where $\thcstressncrho =_{\bydefinition} \frac{\thcstressnc}{\rho}$, while the entropy production takes the form
\begin{equation}
  \label{eq:98}
  \entprodc
  =
  \frac{1}{\temp}
  \left\{
    2 \nu
    \tensordot{\traceless{\gradsym}}{\traceless{\gradsym}}
    +
    \frac{1}{2 \nu_1}
    \left[
      \left(
        1 - \alpha
      \right)
      \tensordot{
        \thcstressnc
        \inverse{
          \left(
            \frac{\thcstressnc}{\mu(\temp)}
            +
            \identity
          \right)
        }
      }
      {
        \thcstressnc
      }
      +
      \alpha
      \tensordot{
        \thcstressnc
      }
      {
        \thcstressnc
      }
    \right]
  \right\}
  +
  \kappa
  \frac{\vectordot{\nabla \temp}{\nabla \temp}}{\temp^2}
  .
\end{equation}

\subsection{Approach based on Gibbs free energy and entropy production maximisation}
\label{sec:entr-prod-maxim}

In a nutshell the hypothesis of entropy production maximisation, see~\cite{rajagopal.kr.srinivasa.ar:on*7}, states that once the entropy production ability of the material is specified in terms of \emph{entropy production}~$\entprodc$ given as a function of thermodynamics affinities/fluxes, and once all other structural constraints are taken into account, then the \emph{constitutive relations for fluxes/affinities} as functions of affinites/fluxes must be chosen in such a way that the entropy production is by this choice of fluxes/affinities maximised. (See also~\cite{malek.j.prusa.v:derivation} for an explanatory presentation and worked out examples.) In the remainder of this section we show how the procedure works for the Giesekus/Oldroyd-B models in the approach based on the Gibbs free energy.

Note however that the terminology affinity/flux, albeit the standard one, might be seen as inappropriate for several reasons, see~\cite{rajagopal.kr.srinivasa.ar:some}. In the current contribution we therefore take the liberty to use the term \emph{resultant} instead of affinity and the term~\emph{determinant} instead of~\emph{flux}, see~\cite{rajagopal.kr.srinivasa.ar:some} for the rationale of this terminology.

\subsubsection{Oldroyd-B model}
\label{sec:oldroyd-b}
The Gibbs free energy remains the same as in the preceding sections, see~\eqref{eq:69}, and the entropy evolution equation remains the same~\eqref{eq:74} as well, that is
\begin{equation}
  \label{eq:99}
  \rho \dd{\entropy}{t}
  +
  \divergence
  \left(
    \frac{\hfluxc}{\temp}
  \right)
  =
  \frac{1}{\temp}
  \left\{
    \tensordot{\left[ \traceless{\cstress} - \devthcstressncdelta \right]}{\traceless{\gradsym}}
    -
    \frac{1}{2}
    \tensordot{
      \thcstressnc
      \exponential{-2 \henckync}
    }
    {
      \fid{
        \overline{
          \exponential{2 \henckync}
        }
      }
    }
  \right\}
  -
  \frac{\vectordot{\hfluxc}{\nabla \temp}}{\temp^2}
  .
\end{equation}
In virtue of thermodynamic relation~$\henckync = - \pd{\gibbs}{\thcstressncrho}(\temp, \thcstressncrho)$ we for the given Gibbs free energy~\eqref{eq:69} get
\begin{equation}
  \label{eq:100}
  \thcstressnc
  =
  \mu (\temp)
  \left(
    \exponential{2 \henckync}
    -
    \identity
  \right)
  ,
\end{equation}
see Section~\ref{sec:rederivation} for details. This relation allows us to rewrite~\eqref{eq:99} as
\begin{equation}
  \label{eq:101}
  \rho \dd{\entropy}{t}
  +
  \divergence
  \left(
    \frac{\hfluxc}{\temp}
  \right)
  =
  \frac{1}{\temp}
  \left\{
    \tensordot{\left[ \traceless{\cstress} - \devthcstressncdelta \right]}{\traceless{\gradsym}}
    -
    \frac{1}{2}
    \tensordot{
      \thcstressnc
      \inverse{
        \left(
          \frac{\thcstressnc}{\mu(\temp)} + \identity
        \right)
      }
    }
    {
      \fid{
        \overline{
          \left(
            \frac{\thcstressnc}{\mu(\temp)} + \identity
          \right)
        }
      }
    }
  \right\}
  -
  \frac{\vectordot{\hfluxc}{\nabla \temp}}{\temp^2}
  .
\end{equation}
This is the entropy production as it is implied by the choice of the Gibbs free energy~\eqref{eq:69}. (Note that~\eqref{eq:101} is unlike~\eqref{eq:99} written down in terms of the stress tensor $\thcstressnc$.) The rate-type quantities and can be identified as resultants,
\begin{subequations}
  \label{eq:102}
  \begin{align}
    \label{eq:103}
    \tensorq{A}_{1} &=_{\bydefinition} \gradsym,
    \\
    \label{eq:104}
    \tensorq{A}_{2} &=_{\bydefinition} \fid{
                      \overline{
                      \left(
                      \frac{\thcstressnc}{\mu(\temp)} + \identity
                      \right)
                      }
                      }
  \end{align}
\end{subequations}
while the corresponding quantities in the dot products in~\eqref{eq:101} can be identified as determinants,
\begin{subequations}
  \label{eq:105}
  \begin{align}
    \label{eq:106}
    \tensorq{J}_1
    &=_{\bydefinition}
      \traceless{\cstress} - \devthcstressncdelta
      ,
    \\
    \label{eq:107}
    \tensorq{J}_2
    &=_{\bydefinition}
      -
      \frac{1}{2}
      \thcstressnc
      \inverse{
      \left(
      \frac{\thcstressnc}{\mu(\temp)} + \identity
      \right)
      }
  \end{align}
\end{subequations}
(We also have the classical determinant--resultant pair $\hfluxc$ and $\nabla \temp$.) On the other hand the assumption regarding the specification of the entropy production is for the Oldroyd-B model the following
\begin{equation}
  \label{eq:108}
  \entprodc
  =
  \frac{1}{\temp}
  \left\{
    2 \nu
    \tensordot{\traceless{\gradsym}}{\traceless{\gradsym}}
    +
    \frac{\nu_1}{2}
    \tensordot{
      \fid{
        \overline{
          \left(
            \frac{\thcstressnc}{\mu(\temp)} + \identity
          \right)
        }
      }
      \inverse{
        \left(
          \frac{\thcstressnc}{\mu(\temp)}
          +
          \identity
        \right)
      }
    }
    {
      \fid{
        \overline{
          \left(
            \frac{\thcstressnc}{\mu(\temp)} + \identity
          \right)
        }
      }
    }
  \right\}
  +
  \kappa
  \frac{\vectordot{\nabla \temp}{\nabla \temp}}{\temp^2}
  ,
\end{equation}
see~\eqref{eq:81} and the subsequent discussion. (We need to rewrite~\eqref{eq:98} in such a way that it contains the resultants identified in~\eqref{eq:102}.) The constitutive relations we need to identify are the constitutive relations for the Cauchy stress tensor $\cstress$ and the stress tensor $\thcstressnc$ and the heat flux $\hfluxc$.

The entropy production maximisation procedure requires us to maximise the entropy production~$\entprodc$ given by~\eqref{eq:108} with respect to resultants~\eqref{eq:102} and $\nabla \temp$, and subject to the constraint
\begin{equation}
  \label{eq:109}
  \entprodc
  =
  \frac{1}{\temp}
  \left\{
    \tensordot{\left[ \traceless{\cstress} - \devthcstressncdelta \right]}{\traceless{\gradsym}}
    -
    \frac{1}{2}
    \tensordot{
      \thcstressnc
      \inverse{
        \left(
          \frac{\thcstressnc}{\mu(\temp)} + \identity
        \right)
      }
    }
    {
      \fid{
        \overline{
          \left(
            \frac{\thcstressnc}{\mu(\temp)} + \identity
          \right)
        }
      }
    }
  \right\}
  -
  \frac{\vectordot{\hfluxc}{\nabla \temp}}{\temp^2}
  ,
\end{equation}
that is implied by the right-hand side of~\eqref{eq:101}. The constrained maximisation problem therefore reads
\begin{equation}
  \label{eq:110}
  \max_{\tensorq{A}_1, \, \tensorq{A}_2, \, \nabla \temp} \entprodc
  \qquad \text{subject to} \qquad
  \entprodc 
  -
  \frac{1}{\temp}
  \left\{
    \tensordot{
      \tensorq{A}_1
    }
    {
      \tensorq{J}_1
    }
    +
    \tensordot{
      \tensorq{A}_2
    }
    {
      \tensorq{J}_2
    }
  \right\}
  +
  \frac{\vectordot{\hfluxc}{\nabla \temp}}{\temp^2}
  =
  0
  ,
\end{equation}
where we have used the determinant/resultant notation introduced in~\eqref{eq:102} and~\eqref{eq:105}. The auxiliary functional for the constrained maximisation problem reads
\begin{equation}
  \label{eq:111}
  \Phi =_{\bydefinition}
  \entprodc
  +
  \lambda
  \left(
    \entprodc 
    -
    \frac{1}{\temp}
    \left\{
      \tensordot{
        \tensorq{A}_1
      }
      {
        \tensorq{J}_1
      }
      +
      \tensordot{
        \tensorq{A}_2
      }
      {
        \tensorq{J}_2
      }
    \right\}
    +
    \frac{\vectordot{\hfluxc}{\nabla \temp}}{\temp^2}
  \right)
  ,
\end{equation}
where $\lambda$ is the Lagrange multiplier. The necessary conditions for the extremum are $\pd{\Phi}{\tensorq{A}_1}=\tensorzero$, $\pd{\Phi}{\tensorq{A}_2}=\tensorzero$ and $\pd{\Phi}{\nabla \temp}=\vec{0}$, which translates to
\begin{subequations}
  \label{eq:112}
  \begin{align}
    \label{eq:113}
    \frac{1+\lambda}{\lambda}\pd{\entprodc}{\tensorq{A}_1}
    &=
      \frac{\tensorq{J}_1}{\temp}
      ,
    \\
    \label{eq:114}
    \frac{1+\lambda}{\lambda}\pd{\entprodc}{\tensorq{A}_2}
    &=
      \frac{\tensorq{J}_2}{\temp}
      ,
    \\
    \label{eq:115}
    \frac{1+\lambda}{\lambda}\pd{\entprodc}{\nabla \temp}
    &=
      -
      \frac{\hfluxc}{\temp^2}
      .
  \end{align}
\end{subequations}

Let us first solve~\eqref{eq:112} for the Lagrange multiplier $\lambda$. Once we identify the value of the Lagrange multiplier $\lambda$, it will be straightforward to find the determinants $\hfluxc$, $\tensorq{J}_1$ and $\tensorq{J}_2$ for which the entropy production attains its maximum.  Taking the dot product of individual equations in~\eqref{eq:112} with the corresponding resultant, and taking the sum of these equations yields
\begin{equation}
  \label{eq:116}
  \frac{1+\lambda}{\lambda}
  \left(
    \tensordot{\pd{\entprodc}{\tensorq{A}_1}}{\tensorq{A}_1}
    +
    \tensordot{\pd{\entprodc}{\tensorq{A}_2}}{\tensorq{A}_2}
    +
    \vectordot{\pd{\entprodc}{\nabla \temp}}{\nabla \temp}
  \right)
  =
  \frac{1}{\temp}
  \left\{
    \tensordot{
      \tensorq{A}_1
    }
    {
      \tensorq{J}_1
    }
    +
    \tensordot{
      \tensorq{A}_2
    }
    {
      \tensorq{J}_2
    }
  \right\}
  -
  \frac{\vectordot{\hfluxc}{\nabla \temp}}{\temp^2}
  .
\end{equation}
Making use of constraint~\eqref{eq:109}, we see that the right-hand side of~\eqref{eq:116} is equal to $\entprodc$. Furthermore, the direct differentiation of $\entprodc$ given by the formula~\eqref{eq:108} yields
\begin{subequations}
  \label{eq:117}
  \begin{align}
    \label{eq:118}
    \pd{\entprodc}{\tensorq{A}_1}
    &=
      \frac{1}{\temp}
      4 \nu \gradsym
      ,
    \\
    \label{eq:119}
    \pd{\entprodc}{\tensorq{A}_2}
    &=
      \frac{1}{\temp}
      \frac{\nu_1}{2}
      \left[
      \fid{
      \overline{
      \left(
      \frac{\thcstressnc}{\mu(\temp)} + \identity
      \right)
      }
      }
      \inverse{
      \left(
      \frac{\thcstressnc}{\mu(\temp)}
      +
      \identity
      \right)
      }
      +
      \inverse{
      \left(
      \frac{\thcstressnc}{\mu(\temp)}
      +
      \identity
      \right)
      }
      \fid{
      \overline{
      \left(
      \frac{\thcstressnc}{\mu(\temp)} + \identity
      \right)
      }
      }
      \right]
      ,
    \\
    \label{eq:120}
    \pd{\entprodc}{\nabla \temp}
    &=
      2
      \kappa
      \frac{\nabla \temp}{\temp^2}
      ,
  \end{align}
\end{subequations}
which implies that
\begin{equation}
  \label{eq:121}
  \tensordot{\pd{\entprodc}{\tensorq{A}_1}}{\tensorq{A}_1}
  +
  \tensordot{\pd{\entprodc}{\tensorq{A}_2}}{\tensorq{A}_2}
  +
  \vectordot{\pd{\entprodc}{\nabla \temp}}{\nabla \temp}
  =
  2 \entprodc
\end{equation}
(This is not surprising since the entropy production is ``quadratic'' in resultants.) Making use of~\eqref{eq:121} in~\eqref{eq:116} allows us to easily identify the Lagrange multiplier $\lambda$. Indeed~\eqref{eq:116} implies that $2 \frac{1+\lambda}{\lambda} \entprodc = \entprodc$, hence
\begin{equation}
  \label{eq:122}
  \frac{1+\lambda}{\lambda} = \frac{1}{2}.
\end{equation}
Having obtained the formula for the Lagrange multiplier, we can go back to~\eqref{eq:112} and find formulas for the determinants. We see that equations~\eqref{eq:112} immediately imply that
\begin{subequations}
  \label{eq:123}
  \begin{align}
    \label{eq:124}
    \traceless{\cstress} - \devthcstressncdelta &= 2 \nu \traceless{\gradsym}, \\
    \label{eq:125}
    \hfluxc &= - \kappa \nabla \temp,
    \\
    \label{eq:126}
    \frac{\nu_1}{4}
    \left[
    \fid{
    \overline{
    \left(
    \frac{\thcstressnc}{\mu(\temp)} + \identity
    \right)
    }
    }
    \inverse{
    \left(
    \frac{\thcstressnc}{\mu(\temp)}
    +
    \identity
    \right)
    }
    +
    \inverse{
    \left(
    \frac{\thcstressnc}{\mu(\temp)}
    +
    \identity
    \right)
    }
    \fid{
    \overline{
    \left(
    \frac{\thcstressnc}{\mu(\temp)} + \identity
    \right)
    }
    }
    \right]
                                                &=
                                                  -
                                                  \frac{1}{2}
                                                  \thcstressnc
                                                  \inverse{
                                                  \left(
                                                  \frac{\thcstressnc}{\mu(\temp)} + \identity
                                                  \right)
                                                  }
                                                  .
  \end{align}
\end{subequations}
If~\eqref{eq:126} holds, then the matrices
$
\inverse{
  \left(
    \frac{\thcstressnc}{\mu(\temp)}
    +
    \identity
  \right)
}
$
and
$
\fid{
  \overline{
    \left(
      \frac{\thcstressnc}{\mu(\temp)} + \identity
    \right)
  }
}
$
commute, which means that~\eqref{eq:126} in fact simplifies to
\begin{equation}
  \label{eq:127}
  \nu_1
  \fid{
    \overline{
      \left(
        \frac{\thcstressnc}{\mu(\temp)} + \identity
      \right)
    }
  }
  =
  -
  \thcstressnc.
\end{equation}
This is the evolution equation for $\thcstressnc$ in the case of Oldroyd-B model, see~\eqref{eq:84} and the discussion following this equation.

In order to show that the matrices
$
\inverse{
  \left(
    \frac{\thcstressnc}{\mu(\temp)}
    +
    \identity
  \right)
}
$
and
$
\fid{
  \overline{
    \left(
      \frac{\thcstressnc}{\mu(\temp)} + \identity
    \right)
  }
}
$
commute, we need to show that they share the same eigenvectors. Let us assume that $\vec{v}$ is an arbitrary eigenvector of
$
\frac{\thcstressnc}{\mu(\temp)}
+
\identity
$
associated to the eigenvalue $\widetilde{\lambda}$. (Since the matrix is symmetric positive definite we know that the normalised eigenvectors form an orthonormal basis, and that the eigenvalues are positive. This also guarantees that the matrix is invertible.) Such an eigenvector is also an eigenvector of matrix
$
\thcstressnc
$
with eigenvalue
$
\mu(\temp)
\left(
  \widetilde{\lambda} - 1
\right)
$%
.
Making use of~\eqref{eq:126} we see that
\begin{equation}
  \label{eq:128}
  \frac{\nu_1}{4}
  \left[
    \fid{
      \overline{
        \left(
          \frac{\thcstressnc}{\mu(\temp)} + \identity
        \right)
      }
    }
    \frac{1}{\widetilde{\lambda}}
    +
    \inverse{
      \left(
        \frac{\thcstressnc}{\mu(\temp)}
        +
        \identity
      \right)
    }
    \fid{
      \overline{
        \left(
          \frac{\thcstressnc}{\mu(\temp)} + \identity
        \right)
      }
    }
  \right]
  \vec{v}
  =
  -
  \frac{1}{2}
  \frac{
    \mu(\temp)
    \left(
      \widetilde{\lambda} - 1
    \right)
  }
  {
    \widetilde{\lambda}
  }
  \vec{v}
  ,
\end{equation}
hence
\begin{equation}
  \label{eq:129}
  \fid{
    \overline{
      \left(
        \frac{\thcstressnc}{\mu(\temp)} + \identity
      \right)
    }
  }
  \vec{v}
  =
  -
  \frac{4}{\nu_1}
  \inverse{
    \left(
      \frac{1}{\widetilde{\lambda}} \identity
      +
      \inverse{
        \left(
          \frac{\thcstressnc}{\mu(\temp)}
          +
          \identity
        \right)
      }
    \right)
  }
  \frac{1}{2}
  \frac{
    \mu(\temp)
    \left(
      \widetilde{\lambda} - 1
    \right)
  }
  {
    \widetilde{\lambda}
  }
  \vec{v}
  ,
\end{equation}
which shows that
$
\fid{
  \overline{
    \left(
      \frac{\thcstressnc}{\mu(\temp)} + \identity
    \right)
  }
}
\vec{v}
=
-
\frac{\mu(\temp)}{\nu_1}
\left(
  \widetilde{\lambda} - 1
\right)
\vec{v}
$%
.
This manipulations shows that an arbitrary eigenvector of 
$
\frac{\thcstressnc}{\mu(\temp)}
+
\identity
$
is also an eigenvector of 
$
\fid{
  \overline{
    \left(
      \frac{\thcstressnc}{\mu(\temp)} + \identity
    \right)
  }
}
$,
which was we set out to prove.

\subsubsection{Giesekus model}
\label{sec:giesekus-1}

The derivation of the Giesekus models proceeds along the same lines as the derivation of the Oldroyd-B model. The \emph{energy storage mechanism} characterised by the postulated Gibbs free energy are the same both for the Oldroyd-B model and the Giesekus model, hence the constraint~\eqref{eq:109} is the same. The only difference is in the postulated \emph{entropy production mechanisms}. Instead of entropy production $\entprodc$ given by~\eqref{eq:108} we now set
\begin{equation}
  \label{eq:130}
  \entprodc
  =
  \frac{1}{\temp}
  \left\{
    2 \nu
    \tensordot{\traceless{\gradsym}}{\traceless{\gradsym}}
    +
    \frac{\nu_1}{2}
    \tensordot{
      \fid{
        \overline{
          \left(
            \frac{\thcstressnc}{\mu(\temp)} + \identity
          \right)
        }
      }
      \inverse{
        \left(
          \frac{\thcstressnc}{\mu(\temp)}
          +
          \identity
        \right)
      }
      \inverse{
        \left[
          \alpha
          \left(
            \frac{\thcstressnc}{\mu(\temp)}
            +
            \identity
          \right)
          +
          \left(
            1-\alpha
          \right)
          \identity
        \right]
      }
    }
    {
      \fid{
        \overline{
          \left(
            \frac{\thcstressnc}{\mu(\temp)} + \identity
          \right)
        }
      }
    }
  \right\}
  +
  \kappa
  \frac{\vectordot{\nabla \temp}{\nabla \temp}}{\temp^2}
  .
\end{equation}

Few observations regarding~\eqref{eq:130} are at hand. First, if we set $\alpha = 0$, then we recover the postulated entropy production~$\entprodc$ for the Oldroyd-B model, see~\eqref{eq:108}. Second, since the matrix 
$
\frac{\thcstressnc}{\mu(\temp)} + \identity
$
is positive definite, and since $\alpha \in [0,1]$, we see that the postulated entropy production is nonnegative. In fact, the critical term in the entropy production has the structure of a weighted dot product of matrices
$
\fid{
  \overline{
    \left(
      \frac{\thcstressnc}{\mu(\temp)} + \identity
    \right)
  }
}
$
that represent the resultant $\tensorq{A}_{2}$, see~\eqref{eq:104}.

The constrained maximisation problem~\eqref{eq:110} is now solved by the same procedure as in the case of Oldroyd-B model. (The identification of determinants and resultants remains the same.) In particular, the ``quadratic'' structure of the entropy production allows one to easily identify the Lagrange multiplier $\lambda$ as a solution to the equation $\frac{1+\lambda}{\lambda} = \frac{1}{2}$. The counterpart of~\eqref{eq:126} is then
\begin{multline}
  \label{eq:131}
  \frac{\nu_1}{4}
  \left[
    \fid{
      \overline{
        \left(
          \frac{\thcstressnc}{\mu(\temp)} + \identity
        \right)
      }
    }
    \inverse{
      \left(
        \frac{\thcstressnc}{\mu(\temp)}
        +
        \identity
      \right)
    }
    \inverse{
      \left[
        \alpha
        \left(
          \frac{\thcstressnc}{\mu(\temp)}
          +
          \identity
        \right)
        +
        \left(
          1-\alpha
        \right)
        \identity
      \right]
    }
  \right.
  \\
  +
    \left.
    \inverse{
      \left(
        \frac{\thcstressnc}{\mu(\temp)}
        +
        \identity
      \right)
    }
    \inverse{
      \left[
        \alpha
        \left(
          \frac{\thcstressnc}{\mu(\temp)}
          +
          \identity
        \right)
        +
        \left(
          1-\alpha
        \right)
        \identity
      \right]
    }
    \fid{
      \overline{
        \left(
          \frac{\thcstressnc}{\mu(\temp)} + \identity
        \right)
      }
    }
  \right]
  =
  -
  \frac{1}{2}
  \thcstressnc
  \inverse{
    \left(
      \frac{\thcstressnc}{\mu(\temp)} + \identity
    \right)
  }
  ,
\end{multline}
and regarding the matrices on the left-hand side, it can be shown again that they commute. Equation~\eqref{eq:131} therefore collapses to
\begin{equation}
  \label{eq:132}
  \nu_1
  \inverse{
    \left(
      \frac{\thcstressnc}{\mu(\temp)}
      +
      \identity
    \right)
  }
  \inverse{
    \left[
      \alpha
      \left(
        \frac{\thcstressnc}{\mu(\temp)}
        +
        \identity
      \right)
      +
      \left(
        1-\alpha
      \right)
      \identity
    \right]
  }
  \fid{
    \overline{
      \left(
        \frac{\thcstressnc}{\mu(\temp)} + \identity
      \right)
    }
  }
=
  -
  \thcstressnc
  \inverse{
    \left(
      \frac{\thcstressnc}{\mu(\temp)} + \identity
    \right)
  }
  ,
\end{equation}
which can be rewritten as
\begin{equation}
  \label{eq:133}
  \nu_1
  \fid{
    \overline{
      \left(
        \frac{\thcstressnc}{\mu(\temp)} + \identity
      \right)
    }
  }
  =
  -
  \thcstressnc
    \left[
      \alpha
      \left(
        \frac{\thcstressnc}{\mu(\temp)}
        +
        \identity
      \right)
      +
      \left(
        1-\alpha
      \right)
      \identity
    \right]
\end{equation}
or as
\begin{equation}
  \label{eq:134}
  \nu_1
  \fid{
    \overline{
      \left(
        \frac{\thcstressnc}{\mu(\temp)} + \identity
      \right)
    }
  }
  =
  -
  \mu(\temp)
  \left(
    \frac{\thcstressnc}{\mu(\temp)}
    +
    \alpha
    \left(
      \frac{\thcstressnc}{\mu(\temp)}
    \right)^2
  \right)
  .
\end{equation}
Simple substitution of~\eqref{eq:67} then reveals that~\eqref{eq:134} is tantamount to the standard evolution equation for~$\lcgnc$ used in the Giesekus model, see~\eqref{eq:giesekus-lcg-evolution-equation} and~\eqref{eq:57}.

\section{Conclusion}
\label{sec:conclusion}

We have shown how a novel approach to the constitutive relations for elastic solids can be embedded in the development of mathematical models for viscoelastic fluids, and we have discussed in detail thermodynamic underpinning of such models. In particular, we have focused on the use of Gibbs free energy instead of the Helmholtz free energy. Using the standard Giesekus/Oldroyd-B models as examples, we have show how the alternative approach works in the case of well-known models and in the fully non-isothermal setting. The proposed approach is straightforward to generalise to more complex setting wherein the classical approach based on the Gibbs potential might be impractical of even inapplicable.

\appendix

\section{Daleckii--Krein formula and its consequences}
\label{sec:daleck-krein-form}

Since we frequently need to differentiate matrix valued functions, it is worth recalling that the derivatives of matrix valued functions are easy to obtain using the Daleckii--Krein formula, see~\cite{daletskii.jl.krein.sg:integration}. (Short proofs of the same are given in~\cite[Theorem~5.3.1]{bhatia.r:positive} and \cite[Theorem~V.3.3]{bhatia.r:matrix}.) Regarding the matrix valued functions and especially tools for computations thereof we refer the reader to~\cite{higham.nj:functions}.

\begin{Theorem}[Daleckii--Krein formula]
  Let $\generictensor$ is a real symmetric matrix in $\R^{k\times k}$ with spectral decomposition $\generictensor = \sum_{i=1}^k \lambda_i \tensorq{P}_i$, where $\left\{\lambda_i\right\}_{i=1}^k$ are the eigenvalues of $\generictensor$, and $\left\{ \tensorq{P}_i \right\}_{i=1}^k$ denotes the projection operator to the $i$-th (normalised) eigenvector $\vec{v}_i$, that is $\tensorq{P}_i =_{\bydefinition} \tensortensor{\vec{v}_i}{\vec{v}_i}$. Let $f$ is a continuously differentiable real function defined on the open set containing the spectrum of~$\generictensor$. Then the corresponding matrix valued function $\tensorf{f}$,
  \begin{equation}
    \label{eq:135}
    \tensorf{f} \left(\generictensor\right) =_{\bydefinition} \sum_{i=1}^k f(\lambda_i) \tensorq{P}_i,
  \end{equation}
  is differentiable, and the Gateaux derivative of $\tensorf{f}$ at point~$\generictensor$ in the direction~$\tensorq{X}$ is given by the formula
  \begin{equation}
    \label{eq:136}
    \Diff[\generictensor] \tensorf{f}(\generictensor) \left[ \tensorq{X} \right]
    =
    \sum_{i=1}^k f^\prime \left( \lambda_i \right)
    \tensorq{P}_i \tensorq{X} \tensorq{P}_i
    +
    \sum_{i=1}^k
    \sum_{\substack{j=1\\j \not= i}}^k
    \frac{f(\lambda_i) - f(\lambda_j)}{\lambda_i - \lambda_j}
    \tensorq{P}_i \tensorq{X} \tensorq{P}_j
    .
  \end{equation}
  Furthermore, let $\tensor{\left[ \tensorschur{\generictensor}{\tensorq{B}} \right]}{_{ij}}=_{\bydefinition} \tensor{\generictensorc}{_{ij}} \tensor{\tensorc{B}}{_{ij}}$ denotes the elementwise Schur/Hadamard product of matrices $\generictensor$ and $\tensorq{B}$. (No summation convention.) Using the Schur product we can rewrite~\eqref{eq:136} as
  \begin{equation}
    \label{eq:137}
    \tensor{
      \left[
        \Diff[\generictensor] \tensorf{f}(\generictensor) \left[ \tensorq{X} \right]
      \right]
    }{_{ij}}
    =
    \frac{f(\lambda_i) - f(\lambda_j)}{\lambda_i - \lambda_j}
    \tensor{\tensorc{X}}{_{ij}}
  \end{equation}
  provided that the Schur product is taken in a basis in which $\generictensor$ is diagonal, and for $i=j$ the quotient in~\eqref{eq:137} is interpreted as $f^\prime(\lambda_i)$.
\end{Theorem}

Note that in the continuum mechanics literature the same formulae have been in special cases rediscovered independently, see form example~\cite{jog.cs:explicit} for the case of matrix logarithm. Using the Daleckii--Krein formula it is straightforward to show the following two lemmas we have been frequently using in our analysis.

\begin{Lemma}
  \label{lemma:differentiation}
  Let $\tensorq{A}$ and $\tensorq{B}$ are real symmetric matrices in $\R^{k\times k}$ that commute, that is $\generictensor \tensorq{B} = \tensorq{B} \generictensor$. Let $f$ be a continuously differentiable real function defined on the open set containing the spectrum of $\generictensor$, and let $\tensorf{f}$ be the matrix valued function associated with $f$, see~\eqref{eq:135}. Then
  \begin{equation}
    \label{eq:138}
    \pd{}{\generictensor}
    \Tr \left[ \tensorq{B} \tensorf{f}(\generictensor)  \right]
    =
    \tensorq{B} \tensorf{f^\prime}(\generictensor)
    ,
  \end{equation}
  where $\tensorf{f}^\prime$ is the matrix valued function associated to $f^\prime$, that is
  \begin{equation}
    \label{eq:139}
    \tensorf{f}^\prime(\generictensor)
    =
    \sum_{i=1}^k f^\prime(\lambda_i) \tensorq{P}_i,
  \end{equation}
  where $\tensorq{A} = \sum_{i=1}^k \lambda_i \tensorq{P}_i$ denotes the spectral decomposition of $\generictensor$.
\end{Lemma}

\begin{proof}
  Since the matrices $\tensorq{A}$ and $\tensorq{B}$ are symmetric matrices and they commute, they share the same set of eigenvectors $\left\{ \vec{v}_i \right\}_{i=1}^3$, and can be simultaneously diagonalised. (They are diagonal in the same basis.) Their spectral decomposition is $\tensorq{A} = \sum_{i=1}^k \lambda_i \tensorq{P}_i$ and $\tensorq{B} = \sum_{i=1}^k \mu_i \tensorq{P}_j$ where $\left\{\lambda_i\right\}_{i=1}^k$ and  $\left\{\mu_i\right\}_{i=1}^k$ denote the eigenvalues of $\generictensor$ and $\tensorq{B}$ respectively. Using the basis wherein the matrices have the diagonal form, the left-hand side of~\eqref{eq:138} reads
  \begin{multline}
    \label{eq:140}
    \tensor{
      \left[
        \pd{}{\generictensor}
        \Tr \left[ \tensorq{B} \tensorf{f}(\generictensor)  \right]
      \right]
    }{_{op}}
    =
    \pd{}{\tensor{\generictensorc}{_{op}}}
    \left[
      \sum_{m,\, n=1}^k
      \mu_m
      \kdelta{_{mn}}
      \tensor{[\tensorf{f}(\generictensor)]}{_{mn}}
    \right]
    =
    \sum_{m,\, n=1}^k
    \mu_m
    \kdelta{_{mn}}
    \pd{\tensor{[\tensorf{f}(\generictensor)]}{_{mn}}}{\tensor{\generictensorc}{_{op}}}
    \\
    =
    \sum_{m,\, n=1}^k
    \mu_m
    \kdelta{_{mn}}
    \frac{f(\lambda_m) - f(\lambda_n)}{\lambda_m - \lambda_n}
    \kdelta{_{mo}}
    \kdelta{_{np}}
    =
    \sum_{o=1}^k
    \mu_o
    f^\prime(\lambda_o)
    \kdelta{_{op}}
    =
    \tensorq{B}
    \tensorf{f}^\prime\left(\generictensor\right)
    .
  \end{multline}
\end{proof}

\begin{Lemma}
  \label{lemma:differentiation-time}
  Suppose that $\tensorq{A}$ and $\tensorq{B}$ are real symmetric matrices in $\R^{k\times k}$ that are differentiable functions with respect to $t \in \R$. Furthermore, suppose that these matrices commute for all $t \in \R$, that is $\generictensor \tensorq{B} = \tensorq{B} \generictensor$, and let $\tensorf{f}$ be a continuously differentiable real function defined on the open set containing the spectrum of $\generictensor$ at every $t \in \R$. Then
  \begin{equation}
    \label{eq:141}
    \Tr
    \left[
      \tensorq{B} \dd{\tensorf{f}(\generictensor)}{t}
    \right]
    =
    \Tr
    \left[
      \tensorq{B} \tensorf{f}^\prime\left(\generictensor\right)
      \dd{\generictensor}{t}
    \right]
    ,
  \end{equation}
  where $\tensorf{f}^\prime$ is the matrix valued function associated to $f^\prime$, see~\eqref{eq:139}.
\end{Lemma}

\begin{proof}
  We use the chain rule, Daleckii--Krein formula and spectral decompostion $\tensorq{A} = \sum_{i=1}^k \lambda_i \tensorq{P}_i$ and $\tensorq{B} = \sum_{i=1}^k \mu_i \tensorq{P}_j$,
  \begin{multline}
    \label{eq:142}
    \Tr
    \left[
      \tensorq{B} \dd{\tensorf{f}(\generictensor)}{t}
    \right]
    =
    \Tr
    \left[
      \tensorq{B} \pd{\tensorf{f}(\generictensor)}{\generictensor} \left[\dd{\generictensor}{t}\right]
    \right]
    \\
    =
    \Tr
    \left[
      \left(
        \sum_{l=1}^k \mu_l \tensorq{P}_l
      \right)
      \left(
        \sum_{i=1}^k f^\prime \left( \lambda_i \right)
        \tensorq{P}_i \dd{\generictensor}{t} \tensorq{P}_i
        +
        \sum_{i=1}^k
        \sum_{\substack{j=1\\j \not= i}}^k
        \frac{f(\lambda_i) - f(\lambda_j)}{\lambda_i - \lambda_j}
        \tensorq{P}_i \dd{\generictensor}{t} \tensorq{P}_j
      \right)
    \right]
    \\
    =
    \Tr
    \left[
      \sum_{i=1}^k
      \mu_i
      f^\prime \left( \lambda_i \right)
      \tensorq{P}_i
      \dd{\generictensor}{t} 
      \tensorq{P}_i
    \right]
    =
    \Tr
    \left[
      \tensorq{B} \tensorf{f}^\prime\left(\generictensor\right)
      \dd{\generictensor}{t}
    \right]
    ,
  \end{multline}
  where we have used the fact that $\tensorq{P}_i \tensorq{P}_j = \kdelta{_{ij}} \tensorq{P}_j$ and the properties of trace. (No summation.)
\end{proof}

Note that the identity does not in general hold without the trace. For example it is not in general true that $\dd{\exponential{\generictensor}}{t} = \exponential{\generictensor} \dd{\generictensor}{t}$.

